\documentclass{lmcs}
\pdfoutput=1

\usepackage{lastpage}
\lmcsdoi{18}{3}{14}
\lmcsheading{}{\pageref{LastPage}}{}{}%
{Jan.~20,~2021}{Aug.~02,~2022}{}

\DeclareFontFamily{U}{shuffle}{}
\DeclareFontShape{U}{shuffle}{m}{n}{<-8>shuffle7 <8->shuffle10}{}

\usepackage[utf8]{inputenc}
\usepackage[T1]{fontenc}
\usepackage[english]{babel}

\usepackage{amsmath}
\usepackage{amssymb}

\usepackage{enumitem}
\usepackage{hyperref}

\newcommand{\displaypunct}[1]{\,\,\text{#1}}
\newcommand{\cal}[1]{\mathcal{#1}}
\newcommand{\set}[1]{\{#1\}}
\newcommand{\card}[1]{\left|#1\right|}
\DeclareMathOperator{\N}{\mathbb{N}}
\DeclareMathOperator{\Omicron}{O}
\DeclareMathOperator{\omicron}{o}
\newcommand{\intinterval}[2]{{[\![}#1, #2{]\!]}}
\newcommand{\function}[5]
{
    #1\colon
    \begin{array}[t]{@{}r@{\,\,}l@{\,\,}l@{}}
    #2 & \to & #3\\
    #4 & \mapsto & #5
    \end{array}
}
\newcommand{\length}[1]{\left|#1\right|}
\newcommand{\emptyword}{\varepsilon}
\newcommand{\quotient}{\mathclose{}/\mathopen{}}
\newcommand{\Reg}{\mathsf{\mathcal{R}eg}}
\newcommand{\FLang}[1]{\mathrm{#1}}
\newcommand{\DLang}[1]{\mathcal{L}\mathopen{}\left(#1\right)\mathclose{}}
\newcommand{\Prog}[1]{\mathcal{P}\mathopen{}\left(#1\right)\mathclose{}}
\newcommand{\NC}[1][\null]
{
    \ifthenelse{\equal{#1}{\null}}
    {
	\mathsf{NC}
    }
    {
	\mathsf{NC^{#1}}
    }
}
\newcommand{\ACC}[1][\null]
{
    \ifthenelse{\equal{#1}{\null}}
    {
	\mathsf{ACC}
    }
    {
	\mathsf{ACC^{#1}}
    }
}
\newcommand{\AC}[1][\null]
{
    \ifthenelse{\equal{#1}{\null}}
    {
	\mathsf{AC}
    }
    {
	\mathsf{AC^{#1}}
    }
}
\newcommand{\FMVariety}[1]{\mathbf{#1}}
\newcommand{\FMVI}{\FMVariety{I}}
\newcommand{\FMVJ}{\FMVariety{J}}
\newcommand{\FMVDA}[1][\null]
{
    \ifthenelse{\equal{#1}{\null}}
    {
	\FMVariety{DA}
    }
    {
	\FMVariety{DA_{#1}}
    }
}
\newcommand{\FMVA}{\FMVariety{A}}
\newcommand{\FMVCom}{\FMVariety{Com}}
\newcommand{\FSVsdprod}{\mathbin{\mathbf{*}}}
\newcommand{\FSVariety}[1]{\mathbf{#1}}
\newcommand{\FSVD}{\FSVariety{D}}
\newcommand{\StVQuasi}{\mathbf{Q}}
\newcommand{\StVQDA}{\StVQuasi\FMVDA}
\newcommand\Mod{\FMVariety{Mod}}
\newcommand\Local{\mathbf{L}}
\newcommand{\StVEsntl}{\mathbf{E}}

\newcommand{\WProb}[1]{\mathcal{W}(#1)}
\newcommand{\nem}{\emph{ne}}
\newcommand{\lmm}{\emph{lm}}
\newcommand{\spv}{\emph{sp}}
\newcommand{\pv}{\emph{p}}
\newcommand{\pr}{\emph{p}}
\let\Oldbot\bot
\renewcommand{\bot}{\mathop{\Oldbot}}
\newcommand\V{\FMVariety{V}}

\newcommand{\EP}[1][\null]
{
    \ifthenelse{\equal{#1}{\null}}
    {
	E\!P
    }
    {
	E\!{#1}
    }
}
\newcommand{\DMon}[1]{\FMVariety{M}(#1)}
\newcommand{\LVariety}[1]{\mathcal{#1}}
\newcommand{\LVSUM}[1][\null]
{
    \ifthenelse{\equal{#1}{\null}}
    {
	\LVariety{SUM}
    }
    {
	\LVariety{SUM}_{#1}
    }
}
\newcommand{\LVSUL}[1][\null]
{
    \ifthenelse{\equal{#1}{\null}}
    {
	\LVariety{SUL}
    }
    {
	\LVariety{SUL}_{#1}
    }
}
\DeclareMathOperator{\alphabet}{alph}

\theoremstyle{plain}\newtheorem{claim}[thm]{Claim}
\newenvironment{lemma}{\begin{lem}}{\end{lem}}
\newenvironment{proposition}{\begin{prop}}{\end{prop}}
\newenvironment{theorem}{\begin{thm}}{\end{thm}}
\newenvironment{definition}{\begin{defi}}{\end{defi}}
\newenvironment{conjecture}{\begin{conj}}{\end{conj}}


\begin{document}
\title[Tameness and the power of programs over monoids in
       \texorpdfstring{$\FMVDA$}{DA}]
      {Tameness and the power of programs\texorpdfstring{\\}{} over monoids in
       \texorpdfstring{$\FMVDA$}{DA}\rsuper*}
\titlecomment{{\lsuper*}Revised and extended version
	      of~\cite{Grosshans-McKenzie-Segoufin-2017} that includes a more
	      inclusive definition of tameness, thus strengthening the statement
	      that $\FMVJ$ is not a tame variety, as explained in
	      Section~\ref{sec:general}.}

\author[N.~Grosshans]{Nathan Grosshans\lmcsorcid{0000-0003-3400-1098}}[a]
\address{Fachbereich Elektrotechnik/Informatik, Universität Kassel, Kassel, Germany}
\email{nathan.grosshans@polytechnique.edu}
\urladdr{https://nathan.grosshans.me}

\author[P.~McKenzie]{Pierre McKenzie}[b]
\address{DIRO, Université de Montréal, Montréal, Canada}
\email{mckenzie@iro.umontreal.ca}

\author[L.~Segoufin]{Luc Segoufin\lmcsorcid{0000-0002-9564-7581}}[c]
\address{Inria, DI ENS, ENS, CNRS, PSL University, Paris, France}
\email{luc.segoufin@inria.fr}

\begin{abstract}
The program-over-monoid model of computation originates with Barrington's
proof that the model captures the complexity class $\NC[1]$.  Here we make progress
in understanding the subtleties of the model. First, we identify a new
tameness condition on a class of monoids that entails a natural
characterization of the regular languages recognizable by programs over
monoids from the class. Second, we prove that the class known as $\FMVDA$
satisfies tameness and hence that the regular languages recognized by
programs over monoids in $\FMVDA$ are precisely those recognizable in the
classical sense by morphisms from $\StVQuasi\FMVDA$.  Third, we show
by contrast that the well studied class of monoids called $\FMVJ$ is not
tame. Finally, we exhibit a program-length-based
hierarchy within the class of languages recognized by programs over monoids
from $\FMVDA$.
 \end{abstract}

\maketitle

\section{Introduction}

A program of range $n$ on alphabet $\Sigma$ over a finite monoid $M$ is a
sequence of pairs $(i,f)$ where $1\leq i\leq n$ and $f:\Sigma\rightarrow M$ is a function.
This program assigns to each word $w_1w_2\cdots w_n$ the monoid
element obtained by multiplying out in $M$ the elements $f(w_i)$, one per pair
$(i,f)$, in the order of the sequence. When an accepting set $F\subseteq M$ is specified,
the program naturally defines the language $L_n$ of
words of length $n$ assigned an element in $F$.  A program
sequence $(P_n)_{n \in \N}$ then defines the language formed by the union of
the~$L_n$.

A program over $M$ is a generalization of a morphism from $\Sigma^*$ to $M$,
and recognition by a morphism equates with acceptance by a finite
automaton. Moving from morphisms to programs has a significant impact on the
expressive power as shown by the seminal result of
Barrington~\cite{Barrington-1989}\footnote{in fact extending the scope of an
  observation made earlier by Maurer and Rhodes~\cite{Maurer-Rhodes-1965}} that
polynomial-length program sequences over the group $S_5$ capture the complexity
class $\NC[1]$ (of languages accepted by bounded fan-in Boolean circuits of
logarithmic depth).

Barrington's result was followed by several results strengthening the
correspondence between circuit complexity and programs over finite monoids. The
classes $\AC[0] \subset \ACC[0] \subseteq \NC[1]$ were characterized by
polynomial-length programs over the aperiodic, the solvable, and all monoids
respectively~\cite{Barrington-1989,Barrington-Therien-1988}.  More generally
for any variety $\V$ of finite monoids one can define the class $\Prog{\V}$ of
languages recognized by polynomial-length programs over a monoid drawn from
$\V$. In particular, if $\FMVA$ is the variety of finite aperiodic monoids, then
$\Prog{\FMVA}$ characterizes the complexity class
$\AC[0]$~\cite{Barrington-Therien-1988}. It was further observed that in a
formal sense, understanding the regular languages of $\Prog{\V}$ is sufficient
to understand the expressive power of $\Prog{\V}$
(see~\cite{McKenzie-Peladeau-Therien-1991}, but also~\cite{Books/Straubing-1994}
for a logical point of view).

In view of the above results it is plausible that algebraic automata theory
methods could help separating complexity classes within $\NC[1]$. But although
partial results in restricted settings were obtained, no breakthrough was
achieved this way.

The reason of course is that programs are much more
complicated than morphisms: programs can read the letter at an input position
more than once, in non-left-to-right order, possibly assigning a different
monoid element each time. This complication can be illustrated with the following example. Consider the
variety of finite monoids known as~$\FMVJ$. This is the variety generated by the
syntactic monoids of all languages defined by the presence or absence of
certain \emph{subwords}, where $u$ is a subword of $v$ if $u$ can be obtained
from $v$ by deleting letters~\cite{Simon-1975}. One deduces that monoids in
$\FMVJ$ are
unable to morphism-recognize the language defined by the regular expression
$(a + b)^* a c^+$. Yet a sequence of programs over a monoid in $\FMVJ$ recognizes
$(a + b)^* a c^+$ by the following trick.  Consider the language $L$ of
all words having $ca$ as a subword but having as subwords neither $cca$, $caa$
nor $cb$. Being defined by the occurrence of subwords, $L$ \emph{is} recognized
by a morphism $\varphi:\{a,b,c\}^*\rightarrow M$ where $M\in \FMVJ$, i.e., for
this $\varphi$ there is an $F\subseteq M$ such that $L = \varphi^{-1}(F)$. Here
is the trick: the program of range $n$ over $M$ given by the sequence of
instructions
\[(2,\varphi),
(1,\varphi),
(3,\varphi),
(2,\varphi),
(4,\varphi),
(3,\varphi),
(5,\varphi),
(4,\varphi), \ldots,
(n,\varphi),
(n-1,\varphi),\] using $F$ as accepting set, defines the set of words of length $n$ in
$(a + b)^* a c^+$.
For instance, on input $abacc$ the program outputs $\varphi(baabcacc)$
which is in $F$, while on inputs $abbcc$ and $abacca$ the program outputs
respectively $\varphi(babbcbcc)$ and $\varphi(baabcaccac)$ which are not in $F$.
(See~\cite[Lemma~4.1]{Grosshans-2020} for a full proof of the fact that
$(a + b)^* a c^+ \in \Prog{\FMVJ}$.)

The first part of our paper addresses the question of what are the regular
languages in $\Prog{\V}$. As mentioned above, this is the key to
understanding the expressive power of $\Prog{\V}$.

Observe first that the class $\DLang{\V}$ of all languages recognized by a morphism
into a monoid in $\V$ is trivially included in $\Prog{\V}$. It turns out that $\Prog{\V}$
always contains more regular languages. For instance, because a program
instruction $(i,f)$ operating on a word $w$ is ``aware'' of~$i$, the program
can have a behavior depending on some arithmetic properties of~$i$. Moreover, a
program's behavior can in general depend on the length of the input words it
treats. So in particular,
as far as regular languages are concerned, a program for a given input
length can take into account the length of the input modulo some fixed number $k$ in
its acceptance set and each program instruction $(i, f)$ can depend on
the value of~$i$ modulo $k$. This can be formalized by assuming
without loss of generality that membership can depend on the length of the word
$w$ at hand modulo a fixed number $k$ and that each letter in $w$ is tagged with its
position modulo $k$. Regular languages recognized this way are exactly the
languages recognized by a stamp (a surjective morphism from $\Sigma^*$ to $M$
with $\Sigma$ an alphabet and $M$ a finite monoid) in the variety of stamps
$\V \FSVsdprod \Mod$, where $\FSVsdprod$ is the wreath product of stamps and
$\Mod$ the variety of cyclic stamps into groups. In other words,
$\DLang{\V\FSVsdprod\Mod}$ is always included in $\Prog{\V}$.

A program over a monoid can also recognize regular languages by
changing its behavior depending on \emph{bounded-length prefixes} and
\emph{suffixes} arbitrarily. To formalize this, we introduce the class
$\StVEsntl\V$ of stamps that, modulo the beginning and the end of a word,
behave essentially like stamps into monoids from $\V$. It is then not too
hard to show that $\DLang{\StVEsntl\V\FSVsdprod\Mod}$ is always included in
$\Prog{\V}$ when $\V$ does contain non-trivial monoids. Many varieties $\V$ are
such that $\Prog{\V}$ cannot recognize more regular languages than those in
$\DLang{\StVEsntl\V\FSVsdprod\Mod}$. This is the case for example of the variety
$\FMVDA$ as we will see below.

Our first result characterizes those varieties $\V$ having the property that
$\Prog{\V}$ does not contain ``many more'' regular languages than does
$\DLang{\StVEsntl\V\FSVsdprod\Mod}$. To this end we introduce the notion of
\emph{tameness} for a variety of finite monoids $\V$
(Definition~\ref{definition-tame}) and our first result shows that a variety of
finite monoids $\V$ is tame if and only if
$\Prog{\V} \cap \Reg \subseteq \DLang{\StVQuasi\StVEsntl\V}$.  Here,
$\DLang{\StVQuasi\V}$ is the class of regular languages recognized by stamps
in quasi-$\V$. A stamp $\varphi$ from $\Sigma^*$ to $M$ is in quasi-$\V$ if,
though $M$ might not be in $\V$, its stable
monoid induced by $\varphi$ is in $\V$, i.e.\ there is a number $k$ such that
$\varphi((\Sigma^k)^*)$ forms a submonoid of $M$ which is in $\V$. For tame
varieties $\V$ we do not know when the inclusion of
$\DLang{\StVEsntl\V\FSVsdprod\Mod}$ in $\DLang{\StVQuasi\StVEsntl\V}$ is
strict or not. In particular we do not know when the inclusion in our result is
an equality. As usual in this context, we conjecture that equality holds at
least for local varieties $\V$.

Our notion of a tame variety differs subtly but fundamentally from the
notion of \pv-variety (program-variety). This notion goes back to
Péladeau~\cite{PhD_thesis/Peladeau} and can be stated by saying that a variety
of finite monoids $\V$ is a \pv-variety whenever any monoid that can be
``simulated'' by programs over a monoid in $\V$ belongs itself to $\V$.
Equivalently, $\V$ is a \pv-variety whenever any regular language in
$\Prog{\V}$ with a neutral letter (a letter which can be inserted and deleted
arbitrarily in words without changing their membership in the language) is in
fact morphism-recognized by a monoid in $\V$. (The equivalence between the two
definitions is claimed without a proof in~\cite{PhD_thesis/Tesson}
and~\cite{PhD_thesis/Grosshans}, see~\cite{Peladeau-Straubing-Therien-1997} for
a proof in one direction, the other direction requiring a standard argument.)
While understanding the neutral letter regular languages in $\Prog{\V}$
for $\V$ ranging over all possible varieties of finite monoids would suffice to
solve most open questions about the internal structure of $\NC[1]$, for $\V$ to
be a \pv-variety does not imply a precise characterization of
\emph{all} the regular languages in $\Prog{\V}$. It can be proved that if $\V$
is a \pv-variety then the regular languages in $\Prog{\V}$ are all in
$\DLang{\StVQuasi\Local\V}$, where $\Local\V$ is the inclusion-wise largest
variety of finite semigroups containing all monoids in $\V$ and only those
monoids. For instance, $\FMVDA$ is a
\pv-variety~\cite{Lautemann-Tesson-Therien-2006}, and this implies that
$\Prog{\FMVDA} \cap \Reg \subseteq \DLang{\StVQuasi\Local\FMVDA}$ as explained
above. The latter inclusion is strict and the correct characterization, namely
$\DLang{\StVQuasi\StVEsntl\FMVDA}$, requires proving that $\FMVDA$ is also tame
in our sense.  Furthermore, there exist \pv-varieties for which unexpected (and
interesting) things happen when considering program-recognition of regular
languages without neutral letter, and this precisely because they aren't tame.
For example, $\Prog{\FMVJ} \cap \Reg \subseteq \DLang{\StVQuasi\Local\FMVJ}$
with strict inclusion while it will follow from our result that
$\Prog{\FMVJ} \cap \Reg \nsubseteq \DLang{\StVQuasi\StVEsntl\FMVJ}$ (knowing
that it is easy to check that
$\DLang{\StVQuasi\StVEsntl\FMVJ} \subseteq \DLang{\StVQuasi\Local\FMVJ}$).

The situation for programs over finite semigroups of the form
$\V \FSVsdprod \FSVD$, where $\V$ is a variety of finite monoids and $\FSVD$ is the
variety of finite definite (or righty trivial) semigroups, turns out to be much
simpler.  Indeed, with the necessary adaptations to the notion of \pv-variety,
Péladeau, Straubing and Thérien~\cite{Peladeau-Straubing-Therien-1997} could
show that for any \pv-variety of the form $\V \FSVsdprod \FSVD$ we have
$\Prog{\V \FSVsdprod \FSVD} \cap \Reg = \DLang{\StVQuasi(\V \FSVsdprod
  \FSVD)}$.  Once our notion of tameness is adapted for finite semigroups, it
is possible to show that any \pv-variety of the form $\V \FSVsdprod \FSVD$ is
tame. Hence the result of~\cite{Peladeau-Straubing-Therien-1997} mentioned
above follows from our result as for varieties of the form
$\V \FSVsdprod \FSVD$ we have, abusing notation, that
$\StVEsntl(\V\FSVsdprod \FSVD)=\V \FSVsdprod \FSVD$ and that
$\DLang{\StVQuasi(\V \FSVsdprod \FSVD)} = \DLang{\V \FSVsdprod \FSVD \FSVsdprod
  \Mod} \subseteq \Prog{\V \FSVsdprod \FSVD}$.  It is to be noted that programs
over semigroups in $\V \FSVsdprod \FSVD$ correspond to Straubing's $k$-programs
over monoids in $\V$~\cite{Straubing-2000, Straubing-2001}. It is also possible
to prove that regular languages recognized by monoids from a $k$-program variety
$\V$, as for
\pv-varieties, are all in $\DLang{\StVQuasi\Local\V}$.  Interestingly, in order
to get a tight characterization for the regular languages recognized by
$k$-programs over commutative monoids in~\cite{Straubing-2001}, Straubing
determines not only which monoids can be simulated by such $k$-programs, but,
in our terms, which \emph{stable stamps} those $k$-programs can
``simulate''. Our notion of tameness also builds upon stable stamps and, as we
have advocated above, subsumes previous definitions of ``good behavior'' of
programs with respect to recognition of regular languages.

Tameness as defined here is also a proper extension of the notion of \spv-varieties
of monoids (Definition~\ref{def:sp-variety}), a concept introduced
in~\cite{Grosshans-McKenzie-Segoufin-2017} as our initial attempt to capture the
expected behavior of programs over small varieties. We will for instance see
that the variety of finite commutative monoids is tame but not an \spv-variety.

Showing that a variety is tame can be a difficult task. For instance
showing that the variety $\FMVA$ is tame amounts to showing that the regular
languages in $\AC[0]$ are in $\DLang{\StVQuasi\FMVA}$, which, as shown by
Barrington, Compton, Straubing and
Thérien~\cite{Barrington-Compton-Straubing-Therien-1992}, follows from the fact
that modulo counting cannot be done in $\AC[0]$ (the famous result initially
proven by Furst, Saxe and Sipser~\cite{Furst-Saxe-Sipser-1984} and independently
by Ajtai~\cite{Ajtai-1983}). Similarly much of
the structure of $\NC[1]$ would in fact be resolved by showing the tameness of
certain varieties (see \cite[Corollary 4.13]{McKenzie-Peladeau-Therien-1991},
\cite[Conjecture IX.3.4]{Books/Straubing-1994}).

The present work is motivated by the need to better understand the subtle
behaviors of polynomial-length programs over monoids. We focus in this paper on
the variety of monoids $\FMVDA$. The importance of $\FMVDA$ in algebraic
automata theory and its connections with other fields are well established
(see~\cite{Tesson-Therien-2002b} for an eloquent testimony). In particular
$\Prog{\FMVDA}$ corresponds to languages accepted by decision trees of bounded
rank~\cite{Gavalda-Therien-2003}. It is also known that regular languages with
a neutral letter that are in $\Prog{\FMVDA}$ are also
in~$\DLang{\FMVDA}$~\cite{Lautemann-Tesson-Therien-2006}.

Our second result shows that the variety
$\FMVDA$ is tame. As it is easy to see that $\FMVDA$ is powerful enough to
describe prefixes and suffixes of words up to some bounded length, we get that
$\DLang{\StVEsntl\FMVDA}=\DLang{\FMVDA}$. Moreover, because $\FMVDA$ is a local
variety, $\StVQuasi\FMVDA=\FMVDA \FSVsdprod \Mod$~\cite{Dartois-Paperman-2013}.
Altogether the tameness of $\FMVDA$ implies that the regular languages in
$\Prog{\FMVDA}$ are precisely the languages in $\DLang{\StVQuasi\FMVDA}$.

Our third result is that, on the other hand, the variety of finite monoids
$\FMVJ$ is not tame as witnessed by the regular language $(a + b)^* a c^+$
discussed above which is in $\Prog{\FMVJ}$ but not in
$\DLang{\StVQuasi\StVEsntl\FMVJ}$. Characterizing the regular languages in
$\Prog{\FMVJ}$ remains an open problem, partially solved
in~\cite{Grosshans-2020}.

Our final result concerns $\Prog{\FMVDA}$. With ${\cal C}_k$ the class of languages
recognized by programs of length $\Omicron(n^k)$ over $\FMVDA$, we prove that
${\cal C}_1\subset {\cal C}_2 \subset \cdots \subset {\cal C}_k \subset \cdots
\subset \Prog{\FMVDA}$ forms a strict hierarchy.  We also relate this hierarchy
to another algebraic characterization of $\FMVDA$ and exhibit conditions on
$M \in \FMVDA$ under which any program over $M$ can be rewritten as an
equivalent subprogram (made of a subsequence of the original sequence of
instructions) of length $\Omicron(n^k)$, refining a result by Tesson and
Thérien~\cite{Tesson-Therien-2002a}.

\medskip

\textbf{Organization of the paper.} In Section~\ref{sec:Preliminaries} we define
programs over monoids, \pr-recognition by such programs and the
necessary algebraic background. The definition of tameness for a variety $\V$ is
given in Section~\ref{sec:general} with our first result showing that regular
languages in $\Prog{\V}$ are included in $\DLang{\StVQuasi\StVEsntl\V}$ if and
only if $\V$ is tame; we also briefly discuss the case of $\FMVJ$, which isn't
tame.
We show that $\FMVDA$ is tame in Section~\ref{sec:caseDA}. Finally,
Section~\ref{sec:hier} contains the hierarchy results about $\Prog{\FMVDA}$.

\section{Preliminaries}
\label{sec:Preliminaries}

This section is dedicated to the introduction of the mathematical material used
throughout this paper. Concerning algebraic automata theory, we only quickly
review the basics and refer the reader to the two classical references of the
domain by Eilenberg~\cite{Books/Eilenberg-1974, Books/Eilenberg-1976} and
Pin~\cite{Books/Pin-1986}.

\paragraph*{General notations.} Let $i, j \in \N$ be two natural numbers. We
shall denote by $\intinterval{i}{j}$ the set of all natural numbers $n \in \N$
verifying $i \leq n \leq j$. We shall also denote by $[i]$ the set
$\intinterval{1}{i}$.

\paragraph*{Words and languages.} Let $\Sigma$ be a finite alphabet. We denote
by $\Sigma^*$ the set of all finite words over $\Sigma$. We also denote by
$\Sigma^+$ the set of all finite non empty words over $\Sigma$, the empty word
being denoted by $\emptyword$. Since all our alphabets and words in this article
are always finite, we shall not mention it anymore from here on.
Given some word $w \in \Sigma$, we denote its length by $\length{w}$ and, for
any $a \in \Sigma$, by $\length{w}_a$ the number of occurrences of the letter
$a$ in $w$.
A \emph{language over $\Sigma$} is a subset of $\Sigma^*$. A language is
\emph{regular} if it can be defined using a regular expression. Given a language
$L$, its \emph{syntactic congruence} $\sim_L$ is the relation on $\Sigma^*$
relating two words $u$ and $v$ whenever for all $x,y \in \Sigma^*$, $xuy \in L$
if and only if $xvy \in L$. It is easy to check that $\sim_L$ is an equivalence
relation and a congruence for concatenation. The \emph{syntactic morphism of
$L$} is the mapping sending any word $u$ to its equivalence class in the
syntactic congruence.

The \emph{quotient of a language $L$ over $\Sigma$ relative to the words $u$
and $v$} is the language, denoted by $u^{-1}Lv^{-1}$, of the words $w$ such that
$uwv \in L$.

\paragraph*{Monoids, semigroups and varieties.} A \emph{semigroup} is a
non-empty set equipped with an associative law that we will write
multiplicatively. A \emph{monoid} is a semigroup with an identity. An example of
a semigroup is $\Sigma^+$, the free semigroup over $\Sigma$. Similarly
$\Sigma^*$ is the free monoid over $\Sigma$. With the exception of free monoids
and semigroups, all monoids and semigroups considered here are finite.
A \emph{morphism $\varphi$ from a semigroup $S$ to a semigroup $T$} is a
function from $S$ to $T$ such that $\varphi(x y) = \varphi(x) \varphi(y)$ for
all $x, y \in S$. A morphism of monoids additionally requires that the identity
is preserved; unless otherwise stated, when we say ``morphism'', we always mean
``monoid morphism''. Any morphism $\varphi\colon \Sigma^* \to M$ for $\Sigma$ an
alphabet and $M$ some monoid is uniquely determined by the images of the letters
of $\Sigma$ by $\varphi$.
A semigroup $T$ is a \emph{subsemigroup} of a semigroup $S$ if $T$ is a subset
of $S$ and is equipped with the restricted law of $S$. Additionally the notion
of submonoids requires the presence of the identity.
A semigroup $T$ \emph{divides} a semigroup $S$ if $T$ is the image by a
semigroup morphism of a subsemigroup of $S$. Division of monoids is defined in
the same way by replacing any occurrence of ``semigroup'' by ``monoid''.
The \emph{Cartesian (or direct) product} of two semigroups is simply the
semigroup given by the Cartesian product of the two underlying sets equipped
with the Cartesian product of their laws.

A language $L$ over $\Sigma$ is \emph{recognized by a monoid $M$} if there is a
morphism $h$ from $\Sigma^*$ to $M$ and a subset $F$ of $M$ such that
$L = h^{-1}(F)$. We also say that \emph{the morphism $h$ recognizes~$L$}.
It is well known that a language is regular if and only if it is recognized by a
finite monoid.
Actually, as $\sim_L$ is a congruence, the quotient $\Sigma^*\!\quotient\!\sim_L$
is a monoid, called \emph{the syntactic monoid of $L$}, that recognizes $L$ via
the syntactic morphism of $L$. The syntactic monoid of $L$ is finite if and only
if $L$ is regular. The quotient $\Sigma^+\!\quotient\!\sim_L$ is analogously
called \emph{the syntactic semigroup of $L$}.

A \emph{variety of finite monoids} is a non-empty class of finite monoids closed
under Cartesian product and monoid division. A \emph{variety of finite
semigroups} is defined similarly. When dealing with varieties, we consider only
varieties of finite monoids or semigroups, so we will drop the
adjective ``finite'' when talking about those.

An element $s$ of a semigroup is \emph{idempotent} if $ss=s$. For any finite
semigroup $S$ there is a positive number (the minimum such number), \emph{the
idempotent power of $S$}, often denoted $\omega$, such that for any element
$s \in S$, $s^\omega$ is idempotent.

A general result of Reiterman~\cite{Reiterman-1982} states that each variety of
monoids (or semigroups) can be defined as the class of all finite monoids
satisfying some set of identities, for an appropriate notion of identity. In our
case, we only use a restricted version of this notion of an identity, that we
understand as a formal equality of terms built on the basis of variables by
using products and $\omega$-powers. A finite monoid is then said to satisfy such
an identity whenever the equality is verified for any setting of the variables
to elements of the monoid, interpreting the $\omega$-power as the idempotent
power of that monoid.
For instance, the variety of finite aperiodic monoids $\FMVA$, known as the
variety of ``group-free'' finite monoids (i.e.\ those verifying that they do not
have any non-trivial group as a subsemigroup), is defined by the identity
$x^\omega = x^{\omega + 1}$.
The variety of monoids $\FMVDA$ is defined by the identity
$(xy)^\omega=(xy)^\omega x (xy)^\omega$.
The variety of monoids $\FMVJ$ is defined by the identities
$(xy)^\omega=(xy)^\omega x=y(xy)^\omega$.
One easily deduces that $\FMVJ \subseteq \FMVDA \subseteq \FMVA$.

\paragraph*{Varieties of languages.}
A \emph{variety of languages} is a class of languages over arbitrary alphabets
closed under Boolean operations, quotients and inverses of morphisms (i.e.\ if
$L$ is a language in the class over an alphabet $\Sigma$, if $\Gamma$ is some
other alphabet and $\varphi\colon \Gamma^* \to \Sigma^*$ is a morphism, then
$\varphi^{-1}(L)$ is also in the class).

Eilenberg showed~\cite[Chapter VII, Section 3]{Books/Eilenberg-1976} that there
is a bijective correspondence between varieties of monoids and varieties of
languages: to each variety of monoids $\V$ we can bijectively associate
$\DLang{\V}$ the variety of languages whose syntactic monoids belong to $\V$
and, conversely, to each variety of languages $\LVariety{V}$ we can bijectively
associate $\DMon{\LVariety{V}}$ the variety of monoids generated by the
syntactic monoids of the languages of $\LVariety{V}$, and these correspondences
are mutually inverse.

When $\V$ is a variety of semigroups, we will denote by $\DLang{\V}$ the class
of languages whose syntactic semigroup belongs to $\V$. There is also an
Eilenberg-type correspondence for an appropriate notion of language varieties,
that is \nem-varieties (non-erasing-varieties) of languages, but we won't
present it here. (The interested reader may have a look at~\cite{Straubing-2002}
as well as~\cite[Lemma~6.3]{Pin-Straubing-2005}.)

\paragraph*{Quasi and locally $\V$ languages, modular counting and predecessor.}
If $S$ is a semigroup we denote by $S^1$ the monoid $S$ if $S$ is already a
monoid and $S \cup \set{1}$ otherwise.

The following definitions are taken from~\cite{Pin-Straubing-2005,
Chaubard-Pin-Straubing-2006b}.
Let $\varphi$ be a surjective morphism from $\Sigma^*$, for $\Sigma$ some
alphabet, to a finite monoid $M$: such a morphism is called a \emph{stamp}.
For all $k$ consider the subset $\varphi(\Sigma^k)$ of $M$. As $M$ is finite
there is a $k$ such that $\varphi(\Sigma^{2k}) = \varphi(\Sigma^k)$. This
implies that $\varphi(\Sigma^k)$ is a semigroup. The semigroup given by the
smallest such $k$ is called the \emph{stable semigroup of $\varphi$} and this
$k$ is called the \emph{stability index} of $\varphi$. If $1$ is the identity of
$M$, then $\varphi(\Sigma^k) \cup \set{1}$ is called the \emph{stable monoid of
$\varphi$}.
If $\V$ is a variety of monoids, then we shall denote by $\StVQuasi{\V}$ the
class of stamps whose stable monoid is in $\V$ and by $\DLang{\StVQuasi{\V}}$
the class of languages whose syntactic morphism is in $\StVQuasi{\V}$.

For $\V$ a variety of monoids, we say that a finite semigroup $S$ is
\emph{locally $\V$} if, for every idempotent $e$ of $S$, the monoid $e S e$
belongs to $\V$; we denote by $\Local\V$ the class of locally-$\V$ finite
semigroups, which happens to be a variety of semigroups.

We now define languages recognized by $\V \FSVsdprod \Mod$ and
$\V \FSVsdprod \FSVD$. We do not use the standard algebraic definition using the
wreath product as we won't need it, but instead give a characterization of the
languages recognized by such algebraic
objects~\cite{Chaubard-Pin-Straubing-2006a,Tilson-1987}.

Let $\V$ be a variety of monoids. We say that a language over $\Sigma$ is in
$\DLang{\V \FSVsdprod \Mod}$ if it is obtained by a finite combination of unions
and intersections of languages over $\Sigma$ for which membership of each word
over $\Sigma$ only depends on its length modulo some integer $k \in \N_{>0}$ and
languages $L$ over $\Sigma$ for which there is a number $k\in\N_{>0}$ and a
language $L'$ over $\Sigma \times \set{0,\ldots,k-1}$ whose syntactic monoid is
in $\V$, such that $L$ is the set of words $w$ that belong to $L'$ after adding
to each letter of $w$ its position modulo $k$.
Observe that neither $\V \FSVsdprod \Mod$ nor $\StVQuasi{\V}$ are varieties of
monoids or semigroups, but classes of stamps that happen to be varieties of
stamps of a certain kind\footnote{To be precise, both are \lmm-varieties of
stamps, as defined in~\cite{Straubing-2002}.}, that we won't introduce.

Similarly we say that a language over $\Sigma$ is in
$\DLang{\V \FSVsdprod \FSVD}$ if it is obtained by a finite combination of
unions and intersections of languages over $\Sigma$ for which membership of each
word over $\Sigma$ only depends on its $k \in \N$ last letters and languages $L$
over $\Sigma$ for which there is a number $k \in \N$ and a language $L'$ over
$\Sigma \times \Sigma^{\leq k}$ (where $\Sigma^{\leq k}$ denotes all
words over $\Sigma$ of length at most $k$) whose syntactic monoid is in $\V$,
such that $L$ is the set of words $w$ that belong to $L'$ after adding to each
letter of $w$ the word composed of the $k$ (or less when near the beginning of
$w$) letters preceding that letter. The variety of semigroups
$\V \FSVsdprod \FSVD$ can then be defined as the one generated by the syntactic
semigroups of the languages in $\DLang{\V \FSVsdprod \FSVD}$ as defined above.

A variety of monoids $\V$ is said to be \emph{local} if
$\DLang{\V \FSVsdprod \FSVD} = \DLang{\Local\V}$. This is not the usual
definition of locality, defined using categories, but it is equivalent to
it~\cite[Theorem 17.3]{Tilson-1987}. One of the consequences of locality that we
will use is that $\DLang{\V \FSVsdprod \Mod} = \DLang{\StVQuasi{\V}}$ when $\V$
is local~\cite[Corollary 18]{Dartois-Paperman-2014}, while
$\DLang{\V \FSVsdprod \Mod} \subseteq \DLang{\StVQuasi{\V}}$ in general
(see~\cite{PhD_thesis/Dartois, PhD_thesis/Paperman}).

\paragraph*{Programs over varieties of monoids.}
Programs over monoids form a non-uniform model of computation, first defined by
Barrington and Thérien~\cite{Barrington-Therien-1988}, extending Barrington's
permutation branching program model~\cite{Barrington-1989}.
Let $M$ be a finite monoid and $\Sigma$ an alphabet. A program $P$ over $M$ is a
finite sequence of instructions of the form $(i, f)$ where $i$ is a positive
integer and $f$ a function from $\Sigma$ to $M$. The \emph{length} of $P$ is the
number of its instructions. A program has \emph{range} $n$ if all its
instructions $(i, f)$ verify $1 \leq i \leq n$. A program $P$ of range $n$
defines a function from $\Sigma^n$, the words of length $n$, to $M$ as follows.
On input $w \in \Sigma^n$, for $w = w_1 \cdots w_n$, each instruction $(i, f)$
outputs the monoid element $f(w_i)$. A sequence of
instructions then yields a sequence of elements of $M$ and their product is the
output $P(w)$ of the program. The only program of range $0$, the empty one,
always outputs the identity of $M$.

A language $L$ over $\Sigma$ is \emph{\pr-recognized} by a sequence of programs
$(P_n)_{n\in\N}$ if for each $n$, $P_n$ has range $n$ and length polynomial in
$n$ and recognizes $L\cap \Sigma^n$, that is, there exists a subset $F_n$ of $M$
such that $L\cap \Sigma^n$ is precisely the set of words $w$ of length $n$ such
that $P_n(w) \in F_n$. In that case, we also say that $L$ is \pr-recognized by
$M$.

We denote by $\Prog{M}$ the class of languages \pr-recognized by a sequence of
programs $(P_n)_{n\in\N}$ over $M$. If
$\FMVariety{V}$ is a variety of monoids we denote by $\Prog{\FMVariety{V}}$ the
union of all $\Prog{M}$ for $M \in \FMVariety{V}$.  

The following is a simple fact about $\Prog{\V}$.
Let $\Sigma$ and $\Gamma$ be two alphabets and $\mu\colon \Sigma^* \to \Gamma^*$
be a morphism. We say that $\mu$ is length multiplying, or that \emph{$\mu$ is
an \lmm-morphism}, if there is a constant $k$ such that for all $a \in \Sigma$,
the length of $\mu(a)$ is $k$.

\begin{lemma}\label{lemma-simple-closure-P}\cite[Corollary 3.5]{McKenzie-Peladeau-Therien-1991}
    For $\V$ any variety of monoids, $\Prog{\V}$ is closed under Boolean
    operations, quotients and inverse images of \lmm-morphisms.  
\end{lemma}

Given two range $n$ programs $P, P'$ over some monoid $M$ using the same input
alphabet $\Sigma$, we shall say that $P'$ is a \emph{subprogram}, a
\emph{prefix} or a \emph{suffix} of $P$ whenever $P'$ is, respectively, a
subword, a prefix or a suffix of $P$, looking at $P$ and $P'$ as words over
$[n] \times M^\Sigma$.

\section{General results about regular languages and programs}
\label{sec:general}

Let $\V$ be a variety of monoids. By definition any regular language recognized
by a monoid in $\V$ is \pr-recognized by a sequence of programs over a monoid in
$\V$.
Actually, since in a program over some monoid in $\V$, the monoid element output
for each instruction can depend on the position of the letter read, hence in
particular on its position modulo some fixed number, it is easy to see that any
regular language in $\DLang{\V \FSVsdprod \Mod}$ is \pr-recognized by a sequence
of programs over some monoid in $\V$. We will see in
Section~\ref{section-essentially} that programs over some monoid in $\V$ can
also \pr-recognize the regular languages that are ``essentially $\V$'' i.e.\ that
differ from a language in $\DLang{\V}$  only on the prefix and suffix of the words.

In this section we characterize those varieties $\V$ such that programs over
monoids in $\V$ do not recognize more regular languages than those mentioned
above.

We first recall the definitions and results around \pv-varieties developed by Péladeau, Tesson,
Straubing and Thérien and then present the definition of \spv-varieties that
was inspired by their work and studied in the conference version of the present
paper.
In order to deal with the limitation of \spv-varieties we then define the
notion of essentially-$\V$ that will be the last ingredient for our definition of
tameness. We then provide an upper bound on the regular languages that can be
\pr-recognized by a sequence of programs over a monoid from a tame variety $\V$.

\subsection{\pv- and \spv-varieties of monoids}
\label{sse:sp-varieties}

We first recall the definition of \pv-varieties. These seem to have been originally defined by
Péladeau in his Ph.D. thesis~\cite{PhD_thesis/Peladeau} and later used by Tesson
in his own Ph.D. thesis~\cite{PhD_thesis/Tesson}. The notion of a \pv-variety
has also been defined for semigroups by Péladeau, Straubing and Thérien
in~\cite{Peladeau-Straubing-Therien-1997}.

Let $\mu$ be a morphism from $\Sigma^*$ to a finite monoid $M$. We denote by
$\WProb{\mu}$ the set of languages $L$ over $\Sigma$ such that $L =
\mu^{-1}(F)$ for some subset $F$ of $M$.  Given a semigroup $S$ there is a
unique morphism $\eta_S\colon S^* \to S^1$ extending the identity on $S$,
called the \emph{evaluation morphism of $S$}. We write $\WProb{S}$ for
$\WProb{\eta_S}$. We define $\WProb{M}$ similarly for any monoid $M$. It is
easy to see that if $M \in \V$ then $\WProb{M} \subseteq \Prog{\V}$. The
condition to be a \pv-variety requires a converse of this observation.

\begin{definition}
    An \emph{\pv-variety of monoids} is a variety $\V$ of monoids such that for
    any finite monoid $M$, if $\WProb{M} \subseteq \Prog{\V}$ then
    $M \in \V$.
\end{definition}

The following result illustrates an important property of \pv-varieties, when
the notion is adapted to varieties of semigroups accordingly.

\begin{proposition}\label{prop-pv}\cite{Peladeau-Straubing-Therien-1997}
    Let $\V\FSVsdprod \FSVD$ be a \pv-variety of semigroups, where $\V$ is a
    variety of monoids.
    
    Then $\Prog{\V \FSVsdprod \FSVD} \cap \Reg = \DLang{\V \FSVsdprod \FSVD
    \FSVsdprod \Mod}$ (where the latter class is defined in the same way as
    $\DLang{\V \FSVsdprod \Mod}$).
\end{proposition}

It is known that $\FMVJ$ is a \pv-variety of monoids~\cite{PhD_thesis/Tesson}
but as we have seen in the introduction, $\Prog{\FMVJ}$ contains languages that are more
complicated than those in $\DLang{\FMVJ\FSVsdprod \Mod}$ (see the end of this
subsection for a proof). In order to capture those varieties for which programs
are well behaved we need a restriction of \pv-varieties and this brings us to
the following definition.

\begin{definition}
\label{def:sp-variety}
    An \emph{\spv-variety of monoids} is a variety $\V$ of monoids such that for
    any finite semigroup $S$, if $\WProb{S} \subseteq \Prog{\V}$ then
    $S^1 \in \V$.
\end{definition}

Hence any \spv-variety of monoids is also a \pv-variety of monoids,
but the converse is not always true as we will see in Proposition~\ref{prop-J-sp-variety} below
 that $\FMVJ$ is not an \spv-variety.

An example of an \spv-variety of monoids is the class of aperiodic monoids
$\FMVA$. This is a consequence of the result that for any number $k > 1$,
checking if $\length{w}_a$ is a multiple of $k$ for $w \in \set{a, b}^*$ cannot
be done in $\AC[0] = \Prog{\FMVA}$~\cite{Furst-Saxe-Sipser-1984,Ajtai-1983}
(we shall denote the corresponding language over the alphabet $\set{0, 1}$ by
$\FLang{MOD_k}$).
Towards a contradiction, assume there would exist a semigroup $S$ such that
$S^1$ is not aperiodic but still $\WProb{S} \subseteq \Prog{\FMVA}$.  Then there
is an $x$ in $S$ such that $x^\omega \neq x^{\omega + 1}$. Consider the morphism
$\mu\colon \set{a, b}^* \to S^1$ sending $a$ to $x^{\omega + 1}$ and $b$ to
$x^\omega$, and the language $L = \mu^{-1}(x^\omega)$. It is easy to see that
$L$ is the language of all words with a number of $a$ congruent to $0$ modulo
$k$, where $k$ is the smallest number such that $x^{\omega + k} = x^\omega$. As
$x^\omega \neq x^{\omega + 1}$, we have $k > 1$, so that $L \notin \Prog{\FMVA}$
by~\cite{Furst-Saxe-Sipser-1984,Ajtai-1983}.
Let $\eta_S\colon S^* \to S^1$ be the evaluation morphism of $S$. The morphism
$\varphi\colon \Sigma^* \to S^*$ sending each letter $a \in \Sigma$ to $\mu(a)$
verifies that $\mu = \eta_S \circ \varphi$, so that
$L = \mu^{-1}(x^\omega) = (\eta_S \circ \varphi)^{-1}(x^\omega) =
 \varphi^{-1}(\eta_S^{-1}(x^\omega))$.
From $\WProb{S} \subseteq \Prog{\FMVA}$ it follows that
$\eta_S^{-1}(x^\omega) \in \Prog{\FMVA}$, hence since $\varphi$ sends each
letter of $\Sigma$ to a letter of $S$, it is an \lmm-morphism and as
$\Prog{\FMVA}$ is closed under inverses of \lmm-morphisms by
Lemma~\ref{lemma-simple-closure-P}, we have
$L = \varphi^{-1}(\eta_S^{-1}(x^\omega)) \in \Prog{\FMVA}$, a contradiction.

The following is the desired consequence of being an \spv-variety of monoids.
\begin{proposition}\label{prop-spv-quasi}
    Let $\V$ be an \spv-variety of monoids.
    Then
    $\Prog{\V} \cap \Reg \subseteq \DLang{\StVQuasi\V}$.
\end{proposition}

\begin{proof}
    Let $L$ be a regular language in $\Prog{M}$ for some $M \in \V$. Let $M_L$
    be the syntactic monoid of $L$ and $\eta_L$ its syntactic morphism. Let $S$
    be the stable semigroup of $\eta_L$, in particular $S = \eta_L(\Sigma^k)$
    for some $k$.
    We wish to show that $S^1$ is in $\V$. 

    We show that $\WProb{S} \subseteq \Prog\V$ and conclude from the fact that
    $\V$ is an \spv-variety that $S^1 \in \V$ as desired.
    Let $\eta_S\colon S^* \to S^1$ be the evaluation morphism of $S$.
    Consider $m \in S$ and consider $L' = \eta_S^{-1}(m)$. We wish to show that
    $L' \in \Prog{\V}$. This implies that $\WProb{S} \subseteq \Prog{\V}$ by
    closure under union, Lemma~\ref{lemma-simple-closure-P}.

    Let $L'' = \eta_L^{-1}(m)$. Since $m$ belongs to the syntactic monoid of $L$
    and $\eta_L$ is the syntactic morphism of $L$, a classical algebraic
    argument~\cite[Chapter 2, proof of Lemma 2.6]{Books/Pin-1986} shows that
    $L''$ is a Boolean combination of quotients of $L$. By
    Lemma~\ref{lemma-simple-closure-P}, we conclude that $L'' \in \Prog{\V}$.

    By definition of $S$, for any element $s$ of $S$ there is a word $u_s$ of
    length $k$ such that $\eta_L(u_s) = s$. Notice that this is precisely where
    we need to work with $S$ and not $S^1$.

    Let $f\colon S^* \to \Sigma^*$ be the \lmm-morphism sending $s$ to $u_s$ and
    notice that $L' = f^{-1}(L'')$. The result follows by closure of $\Prog{\V}$
    under inverse images of \lmm-morphisms, Lemma~\ref{lemma-simple-closure-P}.
\end{proof}

We don't know whether it is always true that for \spv-varieties of monoids $\V$,
$\DLang{\StVQuasi\V}$ is included in $\Prog{\V}$. But we can prove it for local
varieties.

\begin{proposition}\label{prop-local}
\label{ptn:Regular_languages_local_sp-variety_of_finite_monoids}
    Let $\V$ be a local \spv-variety of monoids.
    Then $\Prog{\V} \cap \Reg = \DLang{\StVQuasi\V}$.
\end{proposition}

\begin{proof}
This follows from the fact that for local varieties
$\DLang{\StVQuasi\V} = \DLang{\V\FSVsdprod\Mod}$
(see~\cite{Dartois-Paperman-2014}).
The result can then be derived using Proposition~\ref{prop-spv-quasi}, as we
always have $\DLang{\V\FSVsdprod\Mod} \subseteq \Prog{\V}$.
\end{proof}

As $\FMVA$ is local~\cite[Example 15.5]{Tilson-1987} and an \spv-variety, it
follows from Proposition~\ref{prop-local} that the regular languages in
$\Prog\FMVA$, hence in $\AC[0]$, are precisely those in
$\DLang{\StVQuasi\FMVA}$, which is the characterization of the
regular languages in $\AC[0]$ obtained by Barrington, Compton, Straubing and
Thérien~\cite{Barrington-Compton-Straubing-Therien-1992}.

We will see in the next section that $\FMVDA$ is an \spv-variety. As it is also
local~\cite{Almeida-1996}, we get from Proposition~\ref{prop-local} that the
regular languages of $\Prog{\FMVDA}$ are precisely those in
$\DLang{\StVQuasi\FMVDA}$.

As explained in the introduction, the language  $(a + b)^* a c^+$ can be
\pr-rec\-og\-nized by a program over $\FMVJ$. A simple algebraic argument shows
that it is not in $\DLang{\StVQuasi\FMVJ}$: just compute the stable monoid of
the syntactic morphism of the language, which is equal to the syntactic monoid
of the language, that is not in $\FMVJ$. Hence, by
Proposition~\ref{prop-spv-quasi}, we have the following result:

\begin{proposition}\label{prop-J-sp-variety}
    $\FMVJ$ is not an \spv-variety of monoids.
\end{proposition}

Despite Proposition~\ref{prop-J-sp-variety} providing some explanation for the
unexpected relative strength of programs over monoids in $\FMVJ$, the notion of
an \spv-variety of monoids isn't entirely satisfactory.

We say that a monoid is \emph{trivial} when its underlying set contains a
sole element. The class of all trivial monoids, that we will denote by $\FMVI$,
forms a variety: it is the sole variety containing only trivial monoids, so we
may call it the \emph{trivial variety of monoids}.

One observation to be made is that any non-trivial monoid $M$ \pr-recognizes the
language of words over $\set{a, b}$ starting with an $a$: for the first position
in any word, just send $a$ to any element that is not the identity and $b$ to
the identity. This means that for any non-trivial variety of monoids $\V$, we
have that $a (a + b)^* \in \Prog{\V}$. But since the stable monoid of the
syntactic morphism of $a (a + b)^*$ is equal to the syntactic monoid of this
language, it follows that for any non-trivial variety of monoids $\V$ not
containing the syntactic monoid of $a (a + b)^*$, we have
$\Prog{\V} \cap \Reg \nsubseteq \DLang{\StVQuasi\V}$, hence that $\V$ is not an
\spv-variety of monoids.

Therefore, many varieties of monoids actually aren't \spv-varieties of monoids
simply because of the built-in capacity of programs over any non-trivial monoid
to test the first letter of input words. This is for example true for any
non-trivial variety containing only groups and for any non-trivial variety
containing only commutative monoids. This built-in capacity, additional to
programs' ability to do positional modulo counting that underlies the
definition of \spv-varieties, should be taken into account in the notion we
are looking for to capture ``good behavior''. In order to define our notion of
tameness we first study this extra capacity that is built-in for programs over $\V$ and
that we call ``essentially-$\V$''.

\subsection{Essentially-\texorpdfstring{$\V$}{V} stamps}\label{section-essentially}

It is easy to extend our reasoning above to show that given any
non-trivial monoid $M$ and given some $k \in \N_{>0}$, the language of words
over $\set{a, b}$ having an $a$ in position $k$, that is
$(a + b)^{k - 1} a (a + b)^*$, is \pr-recognized by $M$, and the same goes for
$(a + b)^* a (a + b)^{k - 1}$.  By generalizing, we can quickly conclude that
given any non-trivial variety of monoids $\V$, for any alphabet $\Sigma$ and any
$x, y \in \Sigma^*$, we have that $x \Sigma^* y \in \Prog{\V}$ by closure of
$\Prog{\V}$ under Boolean operations, Lemma~\ref{lemma-simple-closure-P}. Put
informally, \pr-recognition by monoids taken from any fixed non-trivial variety
of monoids allows one to check some constant-length beginning or ending of the
input words. Moreover, \pr-recognition by monoids taken from any fixed
non-trivial variety of monoids $\V$ also easily allows to test for membership of
words in $\DLang{\V}$ after stripping out some constant-length beginning or
ending: that is, languages of the form $\Sigma^{k_1} L \Sigma^{k_2}$ for
$k_1, k_2 \in \N$ and $L \in \DLang{\V}$.

This motivates the definition of \emph{essentially-$\V$} stamps.
\begin{definition}
    Let $\V$ be a variety of monoids.
    Let $\varphi\colon \Sigma^* \to M$ be a stamp and let $s$ be its stability
    index.

    We say that $\varphi$ is \emph{essentially-$\V$} whenever
     there exists a stamp
    $\mu\colon \Sigma^* \to N$ with $N \in \V$ such that for all
    $u, v \in \Sigma^*$, we have
    \[
	\mu(u) = \mu(v) \Rightarrow
	\bigl(\varphi(x u y) = \varphi(x v y) \quad
	      \forall x, y \in \Sigma^s\bigr)
	\displaypunct{.}
    \]
    We will denote by $\StVEsntl\V$ the class of all essentially-$\V$ stamps%
    \footnote{This class actually is an \nem-variety of stamps,
    as defined in~\cite{Straubing-2002}.} and by $\DLang{\StVEsntl\V}$ the
  class of languages recognized by morphisms in  $\StVEsntl\V$.
\end{definition}

Informally stated, a stamp $\varphi\colon \Sigma^* \to M$ is essentially-$\V$
when it behaves like a stamp into a monoid of $\V$ as soon as a
sufficiently long beginning and ending of any input word has been fixed. The
value for ``sufficiently long'' depends on $\varphi$ and is adequately given by
the stability index $s$ of $\varphi$, as by definition of $s$, any word $w$ of
length at least $2 s$ can always be made of length between $s$ and $2 s - 1$
without changing the image by $\varphi$.

Let us start by giving some examples.

Consider first the language $a (a + b)^*$ over the alphabet $\set{a, b}$. Let's
take $\varphi\colon \set{a, b}^* \to M$ to be its syntactic morphism: its
stability index is equal to $1$ and it has the property that for any
$w \in \set{a, b}^*$, we have $\varphi(a w) = \varphi(a)$ and
$\varphi(b w) = \varphi(b)$. Hence, if we define
$\mu\colon \set{a, b}^* \to \set{1}$ to be the obvious stamp into the trivial
monoid $\set{1}$, we indeed have that for all $u, v \in \set{a, b}^*$, it holds
that
\[
    \mu(u) = \mu(v) \Rightarrow
    \bigl(\varphi(x u y) = \varphi(x v y) \quad
	  \forall x, y \in \set{a, b}^1\bigr)
    \displaypunct{.}
\]
In conclusion, the stamp $\varphi$ is essentially-$\V$ for any variety of
monoids, in particular $a(a+b)^* \in \DLang{\StVEsntl\FMVI}$.

\medskip

Let us now consider the language $a (a + b)^* b (a + b)^* a$  over the alphabet
$\set{a, b}$  of words starting and ending with an $a$ and containing some $b$ in between.
Let $\varphi'\colon \set{a, b}^* \to M'$ be its syntactic morphism: its
stability index is equal to $3$ and it has the property that for all
$x, y \in \set{a, b}^+$, given any $u, v \in \set{a, b}^*$ verifying that the
letter $b$ appears in $u$ if and only if it appears in $v$, it holds that
$\varphi'(x u y) = \varphi'(x v y)$. Hence, if we define
$\mu'\colon \set{a, b}^* \to N'$ to be the syntactic morphism of the language
$(a + b)^* b (a + b)^*$, it is direct to see that for all
$u, v \in \set{a, b}^*$, it holds that
\[
    \mu'(u) = \mu'(v) \Rightarrow
    \bigl(\varphi'(x u y) = \varphi'(x v y) \quad
	  \forall x, y \in \set{a, b}^3\bigr)
    \displaypunct{.}
\]
So we can conclude that the stamp $\varphi'$ is essentially-$\V$ for any variety
of monoids containing the syntactic monoid of $(a + b)^* b (a + b)^*$, in
particular $a (a + b)^* b (a + b)^* a \in \DLang{\StVEsntl\FMVJ}$. However, note that
$\varphi' \notin \StVEsntl\FMVI$ because we have
$\varphi'\bigl((aaa) a (aaa)\bigr) \neq \varphi'\bigl((aaa) b (aaa)\bigr)$.

It is now easy to prove that as long as $\V$ is non-trivial, polynomial-length programs
over monoids from $\V$ do have the built-in capacity to recognize any language
recognized by an essentially-$\V$ stamp.

\begin{proposition}
\label{ptn:Closure_P(V)_essentially-V}
    For any non-trivial variety of monoids $\V$, we have
    $\DLang{\StVEsntl\V} \subseteq \Prog{\V}$.
\end{proposition}

\begin{proof}
    Let $\varphi\colon \Sigma^* \to M$ be a stamp in $\StVEsntl\V$. By
    definition, given the stability index $s$ of $\varphi$, there exists a stamp
    $\mu\colon \Sigma^* \to N$ with $N \in \V$ such that for all
    $u, v \in \Sigma^*$, we have
    \[
	\mu(u) = \mu(v) \Rightarrow
	\bigl(\varphi(x u y) = \varphi(x v y) \quad
	      \forall x, y \in \Sigma^s\bigr)
	\displaypunct{.}
    \]

    Let $F \subseteq M$. By definition of $\mu$, given $m \in N$ and
    $x, y \in \Sigma^s$, we either have that
    $x \mu^{-1}(m) y \subseteq \varphi^{-1}(F)$ or that
    $x \mu^{-1}(m) y \cap \varphi^{-1}(F) = \emptyset$. This entails that there
    exist $B \subseteq \Sigma^{\leq 2 s - 1}$ and
    $I \subseteq \Sigma^s \times N \times \Sigma^s$ such that
    \[
	\varphi^{-1}(F) = B \cup \bigcup_{(x, m, y) \in I} x \mu^{-1}(m) y
	\displaypunct{.}
    \]
    We claim that $\set{w} \in \Prog{\V}$ for any
    $w \in \Sigma^{\leq 2 s - 1}$ and also that $x \mu^{-1}(m) y \in \Prog{\V}$
    for any $x, y \in \Sigma^s$ and $m \in N$. So, by closure of $\Prog{\V}$
    under Boolean operations, Lemma~\ref{lemma-simple-closure-P}, it follows
    that $\varphi^{-1}(F) \in \Prog{\V}$.
    Since this is true for any $F$, we have that
    $\WProb{\varphi} \subseteq \Prog{\V}$ and as this is itself true for all
    $\varphi$, we can conclude that $\DLang{\StVEsntl\V} \subseteq \Prog{\V}$.

    The claim remains to be proven.

    Let $k \in \N_{>0}$ and $a \in \Sigma$. Since $\V$ is non-trivial, there
    exists a non-trivial $N' \in \V$: we shall denote its identity by $1$ and by
    $z$ one of its elements distinct from the identity, chosen arbitrarily. It
    is easy to see that the language $\Sigma^{k - 1} a \Sigma^*$ is
    \pr-recognized by the sequence of programs $(P_n)_{n \in \N}$ over $N'$ such
    that for all $n \in \N$, we have
    \[
	P_n =
	\begin{cases}
	    (k, f) & \text{if $n \geq k$}\\
	    \emptyword & \text{otherwise}
	\end{cases}
    \]
    where $f\colon \Sigma \to N'$ is defined by
    $f(b) = \begin{cases}
		z & \text{if $b = a$}\\
		1 & \text{otherwise}
	    \end{cases}$
    for all $b \in \Sigma$. We prove the same for $\Sigma^* a \Sigma^{k - 1}$
    symmetrically.

    It then follows by closure of $\Prog{\V}$ under Boolean operations,
    Lemma~\ref{lemma-simple-closure-P}, that $\set{w} \in \Prog{\V}$ for any
    $w \in \Sigma^{\leq 2 s - 1}$ and that $x \Sigma^* y \in \Prog{\V}$ for any
    $x, y \in \Sigma^s$.

    Finally, let $m \in N$. It is direct to show that there exists
    $L_m \subseteq \Sigma^*$ in $\Prog{\V}$ verifying that
    $L_m \cap \Sigma^s \Sigma^* \Sigma^s = \Sigma^s \mu^{-1}(m) \Sigma^s$: just
    build the sequence of programs $(Q_n)_{n \in \N}$ over $N$ such that for all
    $n \in \N$, we have
    \[
	Q_n =
	\begin{cases}
	    (s + 1, g) (s + 2, g) \cdots (n - s, g) &
		\text{if $n \geq 2 s + 1$}\\
	    \emptyword & \text{otherwise}
	\end{cases}
    \]
    where $g\colon \Sigma \to N$ is defined by $g(b) = \mu(b)$ for all
    $b \in \Sigma$.
    We can then conclude that $x \mu^{-1}(m) y \in \Prog{\V}$ for any
    $x, y \in \Sigma^s$ by closure of $\Prog{\V}$ under Boolean operations,
    Lemma~\ref{lemma-simple-closure-P}, and this holds for any $m$.
\end{proof}

\subsection{Tameness}

We are now ready to define tameness.

We will say that a stamp $\varphi\colon \Sigma^* \to M$ is \emph{stable}
whenever $\varphi(\Sigma^2) = \varphi(\Sigma)$, i.e.\ the stability index of
$\varphi$ is $1$.

\begin{definition}\label{definition-tame}
    A variety of monoids $\V$ is said to be \emph{tame} whenever for any stable
    stamp $\varphi\colon \Sigma^* \to M$, if
    $\WProb{\varphi} \subseteq \Prog{\V}$ then $\varphi \in \StVEsntl\V$.
\end{definition}

Let us first mention that tameness is a generalization of \spv-varieties of monoids.

\begin{proposition}
    Any \spv-variety of monoids is tame.
\end{proposition}

\begin{proof}
    Let $\V$ be an \spv-variety of monoids.

    Let $\varphi\colon \Sigma^* \to M$ be a stable stamp such that
    $\WProb{\varphi} \subseteq \Prog{\V}$.

    Let $S = \varphi(\Sigma^+)$: as $\varphi$ is stable, we have
    $S = \varphi(\Sigma)$. Let $\rho\colon S \to \Sigma$ be an arbitrary mapping
    from $S$ to $\Sigma$ such that $\varphi(\rho(s)) = s$. Consider
    $\eta_S\colon S^* \to S^1$ the evaluation morphism of $S$: the unique
    morphism $f\colon S^* \to \Sigma^*$ sending each letter $s \in S$ to
    $\rho(s)$ verifies that $\eta_S = \varphi \circ f$.
    Now, given any $F \subseteq S^1$, we have
    $\eta_S^{-1}(F) = f^{-1}(\varphi^{-1}(F))$, but since
    $\varphi^{-1}(F) \in \Prog{\V}$ and as $f$ is an \lmm-morphism because it
    sends each letter of $S$ to a letter of $\Sigma$, it follows that
    $\eta_S^{-1}(F) \in \Prog{\V}$ by closure of $\Prog{\V}$ under inverses of
    \lmm-morphisms, Lemma~\ref{lemma-simple-closure-P}.
    Therefore, $\WProb{S} \subseteq \Prog{\V}$.

    Since $\V$ is an \spv-variety of monoids, this entails that
    $M=S^1$ belongs to $\V$, and therefore $\varphi \in \StVEsntl\V$.
    As this is true for any stable stamp $\varphi$ such that
    $\WProb{\varphi} \subseteq \Prog{\V}$, we can conclude that $\V$ is tame.
\end{proof}

The notion of essentially-$\V$ stamps can be adapted to varieties of semigroups
in a straightforward way. We can then define a notion of tameness for varieties
of semigroups accordingly. The exact same proof as the one above then goes
through to allow us to show that \pv-varieties of the form $\V \FSVsdprod \FSVD$
are tame.

However there exist varieties of monoids that are tame but not \spv-varieties. We give an example of
such a variety in Subsection~\ref{section-tame-sp}.

Programs over monoids taken from tame varieties of monoids have the expected
behavior as we show next.

Let $\varphi\colon \Sigma^* \to M$ be a stamp of stability index $s$. The
\emph{stable stamp of $\varphi$} is the unique stamp
$\varphi'\colon (\Sigma^s)^* \to M'$ such that $\varphi'(u) = \varphi(u)$ for
all $u \in \Sigma^s$ and $M'$ is the stable monoid of $\varphi$. For any variety
of monoids $\V$ we let $\StVQuasi\StVEsntl\V$ be the class of stamps whose
stable stamp is essentially-$\V$ and, accordingly, we define
$\DLang{\StVQuasi\StVEsntl\V}$ as the class of languages whose syntactic
morphism is in $\StVQuasi\StVEsntl\V$.

\begin{proposition}\label{prop-tame-quasi-essentially}
    A variety of monoids $\V$ is tame if and only if
    $\Prog{\V} \cap \Reg \subseteq \DLang{\StVQuasi\StVEsntl\V}$.
\end{proposition}

\begin{proof}
    Let $\V$ be a variety of monoids.

    \paragraph{Left-to-right implication.}
    Assume first that $\V$ is tame. For this direction, the proof follows the
    same lines as those of Proposition~\ref{prop-spv-quasi}.

    Let $L \in \Prog{\V} \cap \Reg$ over some alphabet $\Sigma$ and let
    $\eta\colon \Sigma^* \to M$ be the syntactic morphism of $L$.
    For any $m \in M$, a classical algebraic
    argument~\cite[Chapter 2, proof of Lemma 2.6]{Books/Pin-1986} shows that
    $\eta^{-1}(m)$ is a Boolean combination of quotients of $L$, so
    $\eta^{-1}(m) \in \Prog{\V}$ by Lemma~\ref{lemma-simple-closure-P}.

    Now let $s$ be the stability index of $\eta$, let $M'$ be its stable monoid
    and take $\eta'\colon (\Sigma^s)^* \to M'$ to be the stable stamp of $\eta$.
    The unique morphism $f\colon (\Sigma^s)^* \to \Sigma^*$ such that $f(u) = u$
    for all $u \in \Sigma^s$ is an \lmm-morphism and verifies that
    $\eta' = \eta \circ f$. Hence, for all $m' \in M'$, we have that
    $\eta'^{-1}(m') = f^{-1}(\eta^{-1}(m'))$, so that
    $\eta'^{-1}(m') \in \Prog{\V}$ by closure of $\Prog{\V}$ under inverses of
    \lmm-morphisms, Lemma~\ref{lemma-simple-closure-P}.
    Thus, since inverses of monoid morphisms commute with union and $\Prog{\V}$
    is closed under unions (Lemma~\ref{lemma-simple-closure-P}), we can conclude
    that $\eta'^{-1}(F) \in \Prog{\V}$ for all $F \subseteq M'$, i.e.\
    $\WProb{\eta'} \subseteq \Prog{\V}$.

    But as $\eta'$ is stable, by tameness of $\V$, this entails that
    $\eta' \in \StVEsntl\V$, so that $L \in \DLang{\StVQuasi\StVEsntl\V}$.

    \paragraph{Right-to-left implication.}
    Assume now that
    $\Prog{\V} \cap \Reg \subseteq \DLang{\StVQuasi\StVEsntl\V}$.
    Let $\varphi\colon \Sigma^* \to M$ be a stable stamp verifying
    $\WProb{\varphi} \subseteq \Prog{\V}$.

    For any $m \in M$, we therefore have $\varphi^{-1}(m) \in \DLang{\StVQuasi\StVEsntl\V}$. Let
    $\eta_m\colon \Sigma^* \to M_m$ be the syntactic morphism of the language
    $\varphi^{-1}(m)$, we thus have $\eta_m \in \StVQuasi\StVEsntl\V$. We first claim that $\eta_m$ is a stable stamp. To see
    this notice first that for all $u,v\in\Sigma^*$ we have $\varphi(u) =
    \varphi(v) \Rightarrow \eta_m(u) = \eta_m(v)$. Indeed assume that $\varphi(u) =
    \varphi(v)$, then for all  $x, y \in \Sigma^*$ we have $\varphi(x u y) =
    \varphi(x v y)$ hence we have  $xuy \in  \varphi^{-1}(m)$ iff $xvy \in
    \varphi^{-1}(m)$ which entails $\eta_m(u) = \eta_m(v)$ by definition of the
    syntactic morphism. It follows that $\eta_m(\Sigma^2) = \eta_m(\Sigma)$ as $\varphi(\Sigma^2) = \varphi(\Sigma)$.

    Since $\eta_m$ is equal to its stable stamp and $\eta_m \in
    \StVQuasi\StVEsntl\V$, it follows that $\eta_m \in \StVEsntl\V$. Therefore there exists a stamp
    $\mu_m\colon \Sigma^* \to N_m$ with $N_m \in \V$ such that for all
    $u, v \in \Sigma^*$, we have
    \[
	\mu_m(u) = \mu_m(v) \Rightarrow
	\bigl(\eta_m(x u y) = \eta_m(x v y) \quad \forall x, y \in \Sigma\bigr)
	\displaypunct{.}
    \]

    Now, we define the unique stamp $\mu\colon \Sigma^* \to N$ such that
    $\mu(a) = \prod_{m \in M} \mu_m(a)$ for all $a \in \Sigma$ and $N$ is the
    submonoid of $\prod_{m \in M} N_m$ generated by the set
    $\set{\prod_{m \in M} \mu_m(a) \mid a \in \Sigma}$.
    As $\V$ is a variety, $N \in \V$.
    Take $u, v \in \Sigma^*$ and assume $\mu(u) = \mu(v)$: this means that
    $\mu_m(u) = \mu_m(v)$ for all $m \in M$.
    Let $x, y \in \Sigma$. We then have in particular
    $\mu_{\varphi(x u y)}(u) = \mu_{\varphi(x u y)}(v)$. This implies by
    definition of $\mu_{\varphi(x u y)}$ that
    $\eta_{\varphi(x u y)}(x u y) = \eta_{\varphi(x u y)}(x v y)$. As
    $\eta_{\varphi(x u y)}$ is the syntactic morphism of
    $\varphi^{-1}(\varphi(x u y))$, it follows that
    $\varphi(x v y) = \varphi(x u y)$. And this is true for any
    $x, y \in \Sigma$.

    In conclusion, $\mu$ witnesses the fact that $\varphi$ is essentially-$\V$.
\end{proof}

As for the case of \spv-varieties of monoids, we don't know whether it is
always true that for a tame non-trivial variety of monoids $\V$,
$\DLang{\StVQuasi\StVEsntl\V}$ is included in $\Prog{\V}$. If this were the case
then for tame non-trivial varieties of monoids $\V$ we would have $\Prog{\V} \cap \Reg = \DLang{\StVQuasi\StVEsntl\V}$.
We conjecture
this to be at least true for varieties of monoids that are local.

\begin{conjecture}
\label{cjt:Locality_tameness}
    Let $\V$ be a local tame variety of monoids.
    Then $\Prog{\V} \cap \Reg = \DLang{\StVQuasi\StVEsntl\V}$.
\end{conjecture}

We conclude this subsection by showing that $\FMVJ$, which is not an
\spv-variety of monoids (Proposition~\ref{prop-J-sp-variety}), isn't tame
either.

\begin{proposition}
    $\FMVJ$ is not tame.
\end{proposition}

\begin{proof}
    To show this, we show that $(a + b)^* a c^+$, which belongs to
    $\Prog{\FMVJ}$ by the construction of the introduction, does not belong to
    $\DLang{\StVQuasi\StVEsntl\FMVJ}$.

    We first claim that any essentially-$\FMVJ$ stamp
    $\varphi\colon \Sigma^* \to M$ of stability index $s$ verifies that there
    exists some $k \in \N_{>0}$ such that
    $\varphi(x (u v)^k y) = \varphi(x (u v)^k u y)$ for all $u, v \in \Sigma^*$
    and $x, y \in \Sigma^s$.
    Indeed, by definition there exists a stamp $\mu\colon \Sigma^* \to N$ with
    $N \in \FMVJ$ such that for all $u, v \in \Sigma^*$, we have
    \[
	\mu(u) = \mu(v) \Rightarrow
	\bigl(\varphi(x u y) = \varphi(x v y) \quad
	      \forall x, y \in \Sigma^s\bigr)
	\displaypunct{.}
    \]
    If we set $\omega$ to be the idempotent power of $N$, we have that for all
    $u, v \in \Sigma^*$,
    \[
	\mu\bigl((u v)^\omega\bigr) =
	\bigl(\mu(u) \mu(v)\bigr)^\omega =
	\bigl(\mu(u) \mu(v)\bigr)^\omega \mu(u) =
	\mu\bigl((u v)^\omega u\bigr)
    \]
    by the identities for $\FMVJ$.
    Hence, we have that
    $\varphi(x (u v)^\omega y) = \varphi(x (u v)^\omega u y)$ for all
    $u, v \in \Sigma^*$ and $x, y \in \Sigma^s$.

    Let us now consider the syntactic morphism
    $\eta\colon \set{a, b, c}^* \to M$ of the language $(a + b)^* a c^+$. As
    already mentioned for Proposition~\ref{prop-J-sp-variety}, the stable monoid
    of $\eta$ is equal to the syntactic monoid $M$. Moreover, the stability
    index of $\eta$ is $2$. Therefore, the stable stamp of $\eta$ is the unique
    stamp $\eta'\colon (\set{a, b, c}^2)^* \to M$ such that $\eta'(u) = \eta(u)$
    for all $u \in \set{a, b, c}^2$. By what we have shown just above, since the
    stability index of $\eta'$ is $1$, if $\eta'$ were essentially-$\FMVJ$,
    there should exist some $k \in \N_{>0}$ such that
    $\eta'(x (u v)^k y) = \eta'(x (u v)^k u y)$ for all
    $u, v \in (\set{a, b, c}^2)^*$ and $x, y \in \set{a, b, c}^2$.
    However, for all $k \in \N_{>0}$, we do have that
    $(aa) \bigl((bb) (aa)\bigr)^k (cc) \in (a + b)^* a c^+$ while
    $(aa) \bigl((bb) (aa)\bigr)^k (bb) (cc) \notin (a + b)^* a c^+$, which
    implies that
    \[
	\eta'\bigl((aa) \bigl((bb) (aa)\bigr)^k (cc)\bigr) \neq
	\eta'\bigl((aa) \bigl((bb) (aa)\bigr)^k (bb) (cc)\bigr)
    \]
    for all $k \in \N_{>0}$.
    Therefore, it follows that the stable stamp $\eta'$ of $\eta$ is not
    essentially-$\FMVJ$, so we can conclude that
    $(a + b)^* a c^+ \notin \DLang{\StVQuasi\StVEsntl\FMVJ}$.
\end{proof}

\subsection{The example of finite commutative monoids}\label{section-tame-sp}

The variety $\FMVCom$ of finite commutative monoids is defined by the identity
$x y = y x$ and $\DLang{\FMVCom}$ is the class of languages that are Boolean
combinations of languages of the form
$\set{w \in \Sigma^* \mid \length{w}_a \equiv k \mod p}$ for
$k \in \intinterval{0}{p - 1}$ and $p$ prime or
$\set{w \in \Sigma^* \mid \length{w}_a = k}$ for $k \in \N$ with $\Sigma$ any
alphabet and $a \in \Sigma$
(see~\cite[Chapter~VIII, Example~3.5]{Books/Eilenberg-1976}).

Since the syntactic monoid of the language $a (a + b)^*$ is not commutative, by
the discussion at the end of Subsection~\ref{sse:sp-varieties}, we know that
$\FMVCom$ is not an \spv-variety of monoids. It is, however, tame, as we are
going to prove now.

We first give a sufficient equational characterization for any stable stamp
$\varphi$ to be essentially-$\FMVCom$.

\begin{lemma}
    Let $\varphi\colon \Sigma^* \to M$ be a stable stamp verifying that for any
    $x, y, e, f \in \Sigma$ such that $\varphi(e)$ and $\varphi(f)$ are
    idempotents, we have
    \[
	\varphi(e x y f) = \varphi(e y x f)
	\displaypunct{.}
    \]
    Then, $\varphi \in \StVEsntl\FMVCom$.
\end{lemma}

\begin{proof}
    Let us define the equivalence relation $\sim$ on $\Sigma^*$ by $u \sim v$
    for $u, v \in \Sigma^*$ whenever $\varphi(x u y) = \varphi(x v y)$ for all
    $x, y \in \Sigma$.
    This equivalence relation is actually a congruence, because given
    $u, v \in \Sigma^*$ verifying $u \sim v$, for all $s, t \in \Sigma^*$ we
    have $s u t \sim s v t$ since for any $x, y \in \Sigma$, it holds that
    \[
	\varphi(x s u t y) = \varphi(x' u y') = \varphi(x' v y') =
	\varphi(x s v t y)
    \]
    where $x', y' \in \Sigma$ verify $\varphi(x s) = \varphi(x')$ and
    $\varphi(t y) = \varphi(y')$.

    We observe that, since the stability index of $\varphi$ is equal to $1$,
    $\varphi$ verifies that $\varphi(e u v f) = \varphi(e v u f)$ for all
    $u, v \in \Sigma^*$ and $e, f \in \Sigma$ such that $\varphi(e)$ and
    $\varphi(f)$ are idempotents.
    Now take $u, v \in \Sigma^*$ and $x, y \in \Sigma$.
    Since $\varphi(\Sigma)$ is a finite semigroup and verifies that
    $\varphi(\Sigma) = \varphi(\Sigma)^2$, by a classical result in finite
    semigroup theory (see
    e.g.~\cite[Chapter~1, Proposition~1.12]{Books/Pin-1986}), we have that there
    exist $x_1, e, x_2, y_1, f, y_2 \in \Sigma$ such that
    $\varphi(x_1 e x_2) = \varphi(x)$ and $\varphi(y_1 f y_2) = \varphi(y)$ with
    $\varphi(e)$ and $\varphi(f)$ idempotents.
    Therefore, it follows that
    \begin{align*}
	\varphi(x u v y)
	& = \varphi(x_1 e x_2 u v y_1 f y_2)\\
	& = \varphi(x_1 e u v y_1 x_2 f y_2)\\
	& = \varphi(x_1 e u v y_1 x_2 f f y_2)\\
	& = \varphi(x_1 e y_1 x_2 f u v f y_2)\\
	& = \varphi(x_1 e y_1 x_2 f v u f y_2)\\
	& = \varphi(x_1 e v u y_1 x_2 f f y_2)\\
	& = \varphi(x_1 e v u y_1 x_2 f y_2)\\
	& = \varphi(x_1 e x_2 v u y_1 f y_2)\\
	& = \varphi(x v u y)
	\displaypunct{.}
    \end{align*}
    Thus, we have that $u v \sim v u$ for all $u, v \in \Sigma^*$, implying that
    $\Sigma^* \quotient \sim \in \FMVCom$.
    We can eventually conclude that the stamp
    $\mu\colon \Sigma^* \to \Sigma^* \quotient \sim$ defined by
    $\mu(w) = [w]_\sim$ for all $w \in \Sigma^*$ witnesses, by construction, the
    fact that $\varphi$ is essentially-$\FMVCom$.
\end{proof}

The following lemma then asserts that any stable stamp $\varphi$ such that
$\WProb{\varphi} \subseteq \Prog{\FMVCom}$ actually verifies the equation of the
previous lemma, which allows us to conclude that $\FMVCom$ is tame by combining
those two lemmas.

\begin{lemma}\label{lemma-ECOM}
    Let $\varphi\colon \Sigma^* \to M$ be a stable stamp such that
    $\WProb{\varphi} \subseteq \Prog{\FMVCom}$.
    Then, for any $x, y, e, f \in \Sigma$ such that $\varphi(e)$ and
    $\varphi(f)$ are idempotents, we have
    \[
	\varphi(e x y f) = \varphi(e y x f)
	\displaypunct{.}
    \]
\end{lemma}

\begin{proof}
    Let us first observe that for any program $P$ over some finite commutative
    monoid $N$ using the input alphabet $\Sigma$ and of range $n \in \N$, there
    exist a program $P'$ over $N$ using the same input alphabet and of same
    range such that $P' = \prod_{i = 1}^n (i, h_i)$ verifying $P(w) = P'(w)$ for
    all $w \in \Sigma^n$~\cite[Example~3.4]{PhD_thesis/Tesson}. We call $P'$ a
    single-scan program.

    The assumption that $\WProb{\varphi} \subseteq \Prog{\FMVCom}$ thus means
    that for all $F \subseteq M$, there exists a sequence
    $(P_{F, n})_{n \in \N}$ of single-scan programs over some $N_F \in \FMVCom$
    that recognizes $\varphi^{-1}(F)$.

    For all $x, y, e, f, g \in \Sigma$ such that $\varphi(e)$, $\varphi(f)$ and
    $\varphi(g)$ are idempotents, we claim that
    \begin{align}
	\varphi(e x f y g) & = \varphi(e y f x g)\label{eq:Com_1}\\
	\varphi(e f e f) & = \varphi(e f)\label{eq:Com_2}\\
	\varphi(e x y f) & = \varphi(e y e f x f)\label{eq:Com_3}.
    \end{align}

    Assuming the claim, take $x, y, e, f \in \Sigma$ such that $\varphi(e)$ and
    $\varphi(f)$ are idempotents. Equation~\eqref{eq:Com_2} implies that
    $\varphi(e f)$ is an idempotent. As $\varphi(\Sigma^2) = \varphi(\Sigma)$,
    we have that there exists $g \in \Sigma$ such that
    $\varphi(e f) = \varphi(g)$, hence
    \[
	\varphi(e x y f) = \varphi(e y e f x f) = \varphi(e y g x f) =
	\varphi(e x g y f) = \varphi(e x e f y f) = \varphi(e y x f)
	\displaypunct{,}
    \]
    the first and last equalities being from~\eqref{eq:Com_3} and the middle one
    from~\eqref{eq:Com_1}.
      
    Thus, the lemma will be proven once we will have proven
    that~\eqref{eq:Com_1},~\eqref{eq:Com_2} and~\eqref{eq:Com_3} hold for all
    $x, y, e, f, g \in \Sigma$ such that $\varphi(e)$, $\varphi(f)$ and
    $\varphi(g)$ are idempotents.

    \paragraph{\eqref{eq:Com_1} holds.}
    Let $x, y, e, f, g \in \Sigma$ such that $\varphi(e)$, $\varphi(f)$ and
    $\varphi(g)$ are idempotents.
    Set $F = \set{\varphi(e x f y g)}$ and $n = 2 (\card{N_F}^2 + 1) + 1$ and
    assume $P_{F, n} = \prod_{i = 1}^n (i, h_i)$.
    Then, since the function
    \[
	\function{\Delta}{\intinterval{1}{\card{N_F}^2 + 1}}{{N_F}^2}
			 {j}{(h_{2 j}(x), h_{2 j}(y))}
    \]
    cannot be injective, there must exist
    $j_1, j_2 \in \intinterval{1}{\card{N_F}^2 + 1}, j_1 < j_2$ such that
    $h_{2 j_1}(x) = h_{2 j_2}(x)$ and $h_{2 j_1}(y) = h_{2 j_2}(y)$.
    So
    \[
	P_{F, n}(e^{2 j_1 - 1} x f^{2 j_2 - 1 - 2 j_1} y g^{n - 2 j_2}) =
	P_{F, n}(e^{2 j_1 - 1} y f^{2 j_2 - 1 - 2 j_1} x g^{n - 2 j_2})
	\displaypunct{,}
    \]
    hence as
    $e^{2 j_1 - 1} x f^{2 j_2 - 1 - 2 j_1} y g^{n - 2 j_2} \in \varphi^{-1}(F)$
    because
    \[
	\varphi(e^{2 j_1 - 1} x f^{2 j_2 - 1 - 2 j_1} y g^{n - 2 j_2}) =
	\varphi(e x f y g)
	\displaypunct{,}
    \]
    we must have
    $e^{2 j_1 - 1} y f^{2 j_2 - 1 - 2 j_1} x g^{n - 2 j_2} \in \varphi^{-1}(F)$.
    Thus, we have
    \[
	\varphi(e^{2 j_1 - 1} y f^{2 j_2 - 1 - 2 j_1} x g^{n - 2 j_2}) =
	\varphi(e y f x g) = \varphi(e x f y g)
	\displaypunct{.}
    \]

    So for all $x, y, e, f, g \in \Sigma$ such that $\varphi(e)$, $\varphi(f)$
    and $\varphi(g)$ are idempotents, we have that~\eqref{eq:Com_1} holds.

    \paragraph{\eqref{eq:Com_2} holds.}
    Let $e, f \in \Sigma$ such that $\varphi(e)$ and $\varphi(f)$ are
    idempotents. We have that
    \[
	\varphi(e f e f) = \varphi(e f f e f) = \varphi(e e f f f) =
	\varphi(e f)
	\displaypunct{,}
    \]
    the middle equality being from~\eqref{eq:Com_1}.

    So for all $e, f \in \Sigma$ such that $\varphi(e)$ and $\varphi(f)$ are
    idempotents, we have that~\eqref{eq:Com_2} holds.

    \paragraph{\eqref{eq:Com_3} holds.}
    Let $x, y, e, f \in \Sigma$ such that $\varphi(e)$ and $\varphi(f)$ are
    idempotents.
    Set $F = \set{\varphi(e x y f)}$ and $n = 4 (2 \card{N_F}^4 + 1)$ and assume
    $P_{F, n} = \prod_{i = 1}^n (i, h_i)$.
    Then, we have that the function
    \[
	\function{\Delta}
		 {\intinterval{1}{2 \card{N_F}^4 + 1}}{{N_F}^4}
		 {j}{(h_{4 j - 2}(x), h_{4 j - 2}(f), h_{4 j - 1}(y),
		      h_{4 j - 1}(e))}
    \]
    verifies that there exists $(m_1, m_2, m_3, m_4) \in {N_F}^4$ such that
    \[
	\card{\Delta^{-1}((m_1, m_2, m_3, m_4))} \geq 3
	\displaypunct{.}
    \]
    Therefore, there exist
    $j_1, j_2, j_3 \in \intinterval{1}{2 \card{N_F}^4 + 1}, j_1 < j_2 < j_3$
    such that $h_{4 j_3 - 2}(x) = h_{4 j_2 - 2}(x)$,
    $h_{4 j_3 - 2}(f) = h_{4 j_2 - 2}(f)$,
    $h_{4 j_2 - 1}(y) = h_{4 j_1 - 1}(y)$ and
    $h_{4 j_2 - 1}(e) = h_{4 j_1 - 1}(e)$.
    So
    \[
	P_{F, n}(e^{4 j_2 - 3} x y f^{n - 4 j_2 + 1}) =
	P_{F, n}(e^{4 j_1 - 2} y e^{4 (j_2 - j_1) - 2} f e f^{4 (j_3 - j_2) - 2}
		 x f^{n - 4 j_3 + 2})
	\displaypunct{,}
    \]
    hence as $e^{4 j_2 - 3} x y f^{n - 4 j_2 + 1} \in \varphi^{-1}(F)$
    because
    \[
	\varphi(e^{4 j_2 - 3} x y f^{n - 4 j_2 + 1}) = \varphi(e x y f)
	\displaypunct{,}
    \]
    we must have
    $e^{4 j_1 - 2} y e^{4 (j_2 - j_1) - 2} f e f^{4 (j_3 - j_2) - 2}
     x f^{n - 4 j_3 + 2} \in \varphi^{-1}(F)$.
    Thus, we have
    \begin{align*}
	\varphi(e x y f)
	& = \varphi(e^{4 j_1 - 2} y e^{4 (j_2 - j_1) - 2} f e
		    f^{4 (j_3 - j_2) - 2} x f^{n - 4 j_3 + 2})\\
	& = \varphi(e y e f e f x f)\\
	& = \varphi(e y e f x f)
    \end{align*}
    where the last equality uses~\eqref{eq:Com_2}.

    So for all $x, y, e, f \in \Sigma$ such that $\varphi(e)$ and $\varphi(f)$
    are idempotents, we have that~\eqref{eq:Com_3} holds.
\end{proof}

\section{The case of \texorpdfstring{$\FMVDA$}{DA}}
\label{sec:caseDA}

In this section, we prove that $\FMVDA$ is an \spv-variety of monoids, which
implies that it is tame.
Combined with the fact that $\FMVDA$ is local~\cite{Almeida-1996}, we obtain the
following result by
Proposition~\ref{ptn:Regular_languages_local_sp-variety_of_finite_monoids}.

\begin{theorem}
    $\Prog{\FMVDA} \cap \Reg = \DLang{\StVQDA}$.
\end{theorem}

The result follows from the following main technical contribution:

\begin{proposition}\label{prop-main-PDA}
    $(c + ab)^*$, $(b + ab)^*$ and $b^* ((a b^*)^k)^*$ for any
    integer $k \geq 2$ are regular languages not in $\Prog{\FMVDA}$.
\end{proposition}

Observe that for any $k \in \N, k \geq 2$, the fact that
$b^* ((a b^*)^k)^* \notin \Prog{\FMVDA}$ is an immediate corollary of the
classical result that
$\FLang{MOD_k} \notin \AC[0] = \Prog{\FMVA}$~\cite{Furst-Saxe-Sipser-1984, Ajtai-1983}.
However, we propose a direct semigroup-theoretic proof of the first result
without resorting to the involved proof techniques of the latter result.

Before proving the proposition we first show that it implies that $\FMVDA$ is an
\spv-variety of monoids. This implication is a consequence of the following
lemma, which is a result inspired by an observation
in~\cite{Tesson-Therien-2002b} stating that non-membership of a given finite
monoid $M$ in $\FMVDA$ implies non-aperiodicity of $M$ or division of it by (at
least) one of two specific finite monoids.

\begin{lemma}
\label{lem:DA-Witness_languages_non-membership}
    Let $S$ be a finite semigroup such that $S^1 \notin \FMVDA$.
    Then, one of $(c + ab)^*$, $(b + ab)^*$ or $b^* ((a b^*)^k)^*$ for some
    $k \in \N, k \geq 2$ is recognized by a morphism
    $\mu\colon\Sigma^* \to S^1$, for $\Sigma$ the appropriate alphabet, such
    that $\mu(\Sigma^+) \subseteq S$.
\end{lemma}

\begin{proof}
    We distinguish two cases: the aperiodic and the non-aperiodic one.

    \paragraph{Aperiodic case.}
    Assume first that $S^1$ is aperiodic. Then, since $S^1 \notin \FMVDA$, by
    Lemma~3.2.4 in~\cite{MSc_thesis/Tesson}, we have that $S^1$ is divided by
    the syntactic monoid of $(c^* a c^* b c^*)^*$, denoted by $B_2$, or by the
    syntactic monoid of $\bigl((b + c)^* a (b + c)^* b (b + c)^*\bigr)^*$,
    denoted by $U$. We treat those two not necessarily distinct subcases
    separately.

    \subparagraph{Subcase $B_2$ divides $S^1$.}
    It is easily proven that $(c + ab)^*$ is recognized by $B_2$ (actually, its
    syntactic monoid is isomorphic to $B_2$): just consider
    the syntactic morphism $\eta\colon \set{a, b, c}^* \to B_2$ of
    $(c^* a c^* b c^*)^*$ and build the morphism
    $\varphi\colon \set{a, b, c}^* \to B_2$ sending $a$ to $\eta(a)$, $b$ to
    $\eta(b)$ and $c$ to $\eta(ab)$.

    By a classical result in algebraic automata
    theory~\cite[Chapter~1, Proposition~2.7]{Books/Pin-1986}, this implies that
    $S^1$ recognizes $(c + ab)^*$. Thus there exists a morphism
    $\mu\colon \set{a, b, c}^* \to S^1$ recognizing $(c + ab)^*$, that has the
    property that $\mu(\set{a, b, c}^+) \subseteq S$, otherwise there would
    exist $w \in \set{a, b, c}^+$ such that either $w \notin (c + ab)^*$ while
    $\mu(\emptyword) = \mu(w)$ or $w \in (c + ab)^+$ while
    $\mu(a b) = \mu(a w b)$.

    \subparagraph{Subcase $U$ divides $S^1$}
    The proof goes the same way as for the first subcase.

    It is again easily proven that $(b + ab)^*$ is recognized by $U$: here we
    consider the syntactic morphism $\eta\colon \set{a, b}^* \to U$ of
    $\bigl((b + c)^* a (b + c)^* b (b + c)^*\bigr)^*$ and build the morphism
    $\varphi\colon \set{a, b}^* \to U$ sending $a$ to $\eta(a)$ and $b$ to
    $\eta(b)$.

    Thus there exists a morphism $\mu\colon \set{a, b}^* \to S^1$ recognizing
    $(b + ab)^*$, that also verifies $\mu(\set{a, b}^+) \subseteq S$, otherwise
    there would exist $w \in \set{a, b}^+$ such that $w \notin (b + ab)^*$ while
    $\mu(\emptyword) = \mu(w)$ or $w \in b (b + ab)^*$ while
    $\mu(a) = \mu(a w)$ or $w \in ab (b + ab)^*$ while $\mu(a b) = \mu(a w b)$.

    \paragraph{Non-aperiodic case.}
    Assume now that $S^1$ is not aperiodic.
    Then there is an $x$ in $S$ such that $x^\omega \neq x^{\omega + 1}$ for
    $\omega \in \N_{>0}$ the idempotent power of $S^1$.
    Consider the morphism $\mu\colon \set{a, b}^* \to S^1$ sending $a$ to
    $x^{\omega + 1}$ and $b$ to $x^\omega$, and the language
    $L = \mu^{-1}(x^\omega)$.
    Let $k \in \N, k \geq 2$ be the smallest positive integer such that
    $x^{\omega + k} = x^\omega$, that cannot be $1$ because
    $x^\omega \neq x^{\omega + 1}$.
    Using this, for all $w \in \set{a, b}^*$, we have
    \[
	\mu(w) = x^{\length{w} \cdot \omega + \length{w}_a} =
	x^{\omega + (\length{w}_a \mod k)}
	\displaypunct{,}
    \]
    where $\length{w}$ indicates the length of $w$ and $\length{w}_a$ the number
    of $a$'s it contains, so that $w$ belongs to $L$ if and only if
    $\length{w}_a = 0 \mod k$. Hence, $L$ is the language of all words with a
    number of $a$'s divisible by $k$, $b^* ((a b^*)^k)^*$. In conclusion,
    $b^* ((a b^*)^k)^*$ is recognized by $\mu$ verifying
    $\mu(\set{a, b}^+) \subseteq S$.
\end{proof}

Let now $S$ be any finite semigroup such that
$\WProb{S} \subseteq \Prog{\FMVDA}$.
Let $\eta_S\colon S^* \to S^1$ be the evaluation morphism of $S$.
To show that $S^1$ is in $\FMVDA$, we assume for the sake of contradiction that
it is not the case. Then Lemma~\ref{lem:DA-Witness_languages_non-membership}
tells us that one of $(c + ab)^*$, $(b + ab)^*$ or $b^* ((a b^*)^k)^*$ for some
$k \in \N, k \geq 2$ is recognized by a morphism $\mu\colon\Sigma^* \to S^1$,
for $\Sigma$ the appropriate alphabet, such that $\mu(\Sigma^+) \subseteq S$.

In all cases, we thus have a language $L \subseteq \Sigma^*$ equal to
$\mu^{-1}(Q)$ for some subset $Q$ of $S^1$ with the morphism $\mu$ sending
letters of $\Sigma$ to elements of $S$. Consider then the morphism
$\varphi\colon \Sigma^* \to S^*$ sending each letter $a \in \Sigma$ to $\mu(a)$,
a letter of $S$: we have $\mu = \eta_S \circ \varphi$, so that
$L = \varphi^{-1}(\eta_S^{-1}(Q))$.
As $\WProb{S} \subseteq \Prog{\FMVDA}$, we have that
$\eta_S^{-1}(Q) \in \Prog{\FMVDA}$, hence since $\varphi$ is an \lmm-morphism
and $\Prog{\FMVDA}$ is closed under inverses of \lmm-morphisms by
Lemma~\ref{lemma-simple-closure-P}, we have
$L = \varphi^{-1}(\eta_S^{-1}(Q)) \in \Prog{\FMVDA}$: a contradiction to
Proposition~\ref{prop-main-PDA}.

In the remaining part of this section we prove Proposition~\ref{prop-main-PDA}.

\paragraph*{Proof of Proposition~\ref{prop-main-PDA}.}
The idea of the proof is the following. We work by contradiction and assume that
we have a sequence of programs over some monoid $M$ of $\FMVDA$ deciding one of
the targeted language $L$. Let $n$ be much larger than the size of $M$, and let
$P_n$ be the program running on words of length $n$. Consider a set $\Delta$ of
words such that $L \subseteq \Delta^*$ (for instance take
$\Delta = \set{c, ab}$ for $L = (c + ab)^*$). We will show that we
can fix a constant (depending on $M$ and $\Delta$ but not on $n$) number of
entries to $P_n$ such that $P_n$ always outputs the same value and there are
completions of the entries in $\Delta^*$. Hence, if $\Delta$ was chosen so
that there is actually a completion of the fixed entries in
$L$ and one outside of $L$, $P_n$ cannot recognize the restriction of $L$ to
words of length $n$. We cannot prove this
for all $\Delta$, in particular it will not work for $\Delta=\set{ab}$ and
indeed $(ab)^*$ is in $\Prog{\FMVDA}$. The key property of our $\Delta$ is that
after fixing any letter at any position, except maybe for a constant number of
positions, one can still complete the word into one within $\Delta^*$. This is
not true for $\Delta = \set{ab}$ because after fixing a $b$ in an odd position
all completions fall outside of $(ab)^*$.

We now spell out the technical details.

Let $\Delta$ be a finite non-empty set of non-empty words over an alphabet
$\Sigma$. Let $\bot$ be a letter not in $\Sigma$. A \emph{mask}
is a word over $\Sigma\cup \set{\bot}$. The positions of a mask carrying a
$\bot$ are called \emph{free} while the positions carrying a letter in $\Sigma$
are called \emph{fixed}. A mask $\lambda'$ is a \emph{submask} of a mask
$\lambda$ if it is formed from $\lambda$ by replacing some occurrences (possibly
zero) of $\bot$ by a letter in $\Sigma$.

A \emph{completion} of a mask $\lambda$ is a word $w$ over $\Sigma$ that is
built from $\lambda$ by replacing all occurrences of $\bot$ by a letter in
$\Sigma$. Notice that all completions of a mask have the same length as the
mask itself.  A mask $\lambda$ is \emph{$\Delta$-compatible} if it has a
completion in $\Delta^*$.

The \emph{dangerous} positions of a mask $\lambda$ are the positions within
distance $2 l - 2$ of the fixed positions or within distance $l - 1$ of the
beginning or the end of the mask, where $l$ is the maximal length of a word in
$\Delta$. A position that is not dangerous is said to be \emph{safe} and is
necessarily free.

We say that $\Delta$ is \emph{safe} if the following holds. Let $\lambda$ be a
$\Delta$-compatible mask. Let $i$ be any free position of $\lambda$ that is not
dangerous. Let $a$ be any letter in $\Sigma$. Then the submask of $\lambda$
constructed by fixing $a$ at position $i$ is $\Delta$-compatible. We have
already seen that $\Delta=\set{ab}$ is not safe. However our targeted $\Delta$,
$\Delta=\set{c,ab}$, $\Delta=\set{b,ab}$, $\Delta=\set{a,b}$, are safe. We always
consider $\Delta$ to be safe in the following.

Note that it is important in the definition of safe for $\Delta$ that we
fix only safe positions, i.e.\ positions far apart and far from the beginning
and the end of the mask. Indeed, depending on the chosen $\Delta$, there might
be words that never appear as factors in any word of $\Delta^*$, such as $bb$
when $\Delta = \set{c, ab}$ or $aa$ when $\Delta = \set{b, ab}$, so that fixing
a position near an already fixed position to an arbitrary letter in a
$\Delta$-compatible mask may result in a mask that has no completion in
$\Delta^*$. This is why we make sure that safe positions are far from those
already fixed and from the beginning and the end of the mask, where far depends
on the length of the words of $\Delta$.

Finally, we say that a completion $w$ of a mask $\lambda$ is \emph{safe} if $w$
is a completion of $\lambda$ belonging to $\Delta^*$ or is constructed from a
completion of $\lambda$ in $\Delta^*$ by modifying only letters at safe
positions of $\lambda$, the dangerous positions remaining unchanged.

\medskip

Let $M$ be a monoid in $\FMVDA$ whose identity we will denote by $1$.

We define a version of Green's relations for decomposing monoids that will be
used, as often in this setting, to prove the main technical lemma in the current
proof.
Given two elements $u,u'$ of $M$ we say that $u \leq_J u'$ if there are
elements $v,v'$ of $M$ such that $u'=vuv'$. We write $u \sim_J u'$ if $u \leq_J
u'$ and $u' \leq_J u$. We write $u <_J u'$ if $u \leq_J
u'$ and $u' \not\sim_J u$.
Given two elements $u,u'$ of $M$ we say that $u \leq_R u'$ if there is
an element $v$ of $M$ such that $u'=uv$.  We write $u \sim_R u'$ if $u \leq_R
u'$ and $u' \leq_R u$. We write $u <_R u'$ if $u \leq_R
u'$ and $u' \not\sim_R u$.
Given two elements $u,u'$ of $M$ we say that $u \leq_L u'$ if there is
an element $v$ of $M$ such that $u'=vu$. We write $u \sim_L u'$ if $u \leq_L
u'$ and $u' \leq_L u$. We write $u <_L u'$ if $u \leq_L
u'$ and $u' \not\sim_L u$.
Finally, given two elements $u,u'$ of $M$, we write $u \sim_H u'$ if
$u \sim_R u'$ and $u \sim_L u'$.

We shall use the following well-known fact about these preorders and equivalence
relations (see~\cite[Chapter 3, Proposition 1.4]{Books/Pin-1986}).

\begin{lemma}
\label{lem:Green's_relations_link}
    For all elements $u$ and $v$ of $M$, if $u \leq_R v$ and $u \sim_J v$, then
    $u \sim_R v$. Similarly, if $u \leq _L v$ and $u \sim_J v$, then
    $u \sim_L v$.
\end{lemma}

From the definition it follows that for all elements $u,v,r$ of $M$, we have $u
\leq_R ur$ and $v \leq_L rv$. When the inequality is strict in the first case,
i.e.\ $u <_R ur$, we say that $r$ is $R$-bad for $u$. Similarly $r$ is $L$-bad
for $v$ if $v <_L rv$. It follows from $M\in \FMVDA$ that being $R$-bad or
$L$-bad only depends on the $\sim_R$ or $\sim_L$ class, respectively.
This is formalized in the following lemma, that is folklore and used at least
implicitly in many proofs involving $\FMVDA$ (see for
instance~\cite[proof of Theorem 3]{Tesson-Therien-2002b}). Since we didn't
manage to find the lemma stated and proven in the form below, we include a proof
for completeness.

\begin{lemma}
\label{lem:Green's_relations_DA}
    If $M$ is in $\FMVDA$, then $u\sim_R u'$ and $ur \sim_R u$ implies
    $u'r \sim_R u$. Similarly $u\sim_L u'$ and $ru \sim_L u$ implies
    $ru' \sim_L u$.
\end{lemma}

\begin{proof}
    Let $u, u', r \in M$ such that $u \sim_R u'$ and $u r \sim_R u$. This means
    that there exist $v, v', s \in M$ such that $u = u' v'$, $u' = u v$ and
    $u r s = u$.

    This implies that
    \[
	u' = u v = u r s v = u' v' r s v = u' (v' r s v)^2 = \cdots =
	u' (v' r s v)^\omega
    \]
    where $\omega$ is the idempotent power of $M$.
    Hence, we have that
    \[
	u' r (v' r s v)^\omega v' = u' (v' r s v)^\omega r (v' r s v)^\omega v'
	\displaypunct{.}
    \]
    But, by~\cite[Theorem 2]{Tesson-Therien-2002b}, since $M \in \FMVDA$, we
    have that $(xyz)^\omega y (xyz)^\omega = (xyz)^\omega$ for all
    $x, y, z \in M$, so that
    \[
	u' r (v' r s v)^\omega v' = u' (v' r s v)^\omega v' = u' v' = u
	\displaypunct{.}
    \]
    Therefore, we have $u' r \leq_R u$ and since $u v r = u' r$, we also have
    $u \leq_R u' r$, so that $u' r \sim_R u$ as claimed.

    The proof goes through symmetrically for $\sim_L$.
\end{proof}

Let $\Delta$ be a finite set of words and $\Sigma$ be the corresponding
alphabet, $\Delta$ being safe, and let $n \in \N$.
We are now going to prove the main technical lemma that allows us to assert that
after fixing a constant number of positions in the input of a program over $M$,
it can still be completed into a word of $\Delta^*$, but the program cannot make
the difference between any two possible completions anymore. To prove the lemma,
we define a relation $\prec$ on the set of quadruplets $(\lambda, P, u, v)$
where $\lambda$ is a mask of length $n$, $P$ is a program over $M$ for words of
length $n$ and $u$ and $v$ are two elements of $M$. We will say that an element
$(\lambda_1, P_1, u_1, v_1)$ is strictly smaller than
$(\lambda_2, P_2, u_2, v_2)$, written
$(\lambda_1, P_1, u_1, v_1) \prec (\lambda_2, P_2, u_2, v_2)$, if and only if
$\lambda_1$ is a submask of $\lambda_2$, $P_1$ is a subprogram of $P_2$ and one
of the following cases occurs:
\begin{enumerate}
    \item
	$u_2 <_R u_1$ and $v_1 = v_2$ and $P_1$ is a suffix of $P_2$ and
	$u_1 P_1(w) v_1 = u_2 P_2(w) v_2$ for all safe completions $w$ of
	$\lambda_1$;
    \item
	$v_2 <_L v_1$ and $u_1 = u_2$ and $P_1$ is a prefix of $P_2$ and
	$u_1 P_1(w) v_1 = u_2 P_2(w) v_2$ for all safe completions $w$ of
	$\lambda_1$;
    \item
	$u_2 = u_1$ and $v_1 = 1$ and $P_1$ is a prefix of $P_2$ and
	$u_1 P_1(w) v_1 <_J u_2 P_2(w) v_2$ for all safe completions $w$ of
	$\lambda_1$;
    \item
	$v_2 = v_1$ and $u_1 = 1$ and $P_1$ is a suffix of $P_2$ and
	$u_1 P_1(w) v_1 <_J u_2 P_2(w) v_2$ for all safe completions $w$ of
	$\lambda_1$.
\end{enumerate}
Note that, since $M$ is finite, this relation is well-founded (that is, it has
no infinite decreasing chain, an infinite sequence of quadruplets
$\mu_0, \mu_1, \mu_2, \ldots$ such that $\mu_{i + 1} \prec \mu_i$ for all
$i \in \N$) and the maximal length of any decreasing chain can be upper bounded
by $2 \cdot \card{M}^2$, that does only depend on $M$. For a given quadruplet
$\mu$, we shall also call its height the biggest $i \in \N$ such that there
exists a decreasing chain
$\mu_i \prec \mu_{i - 1} \prec \cdots \prec \mu_0 = \mu$.

The following lemma is the key to the proof. It shows that modulo fixing a few
entries, one can fix the output: to count the number of fixed positions for a
given mask $\lambda$, we denote by $\length{\lambda}_\Sigma$ the number of
letters in $\lambda$ belonging to $\Sigma$, that is to say, the number of fixed
positions in $\lambda$.

\begin{lemma}
\label{lemma-R}
    Let $\lambda$ be a $\Delta$-compatible mask of length $n$, let $P$ be a
    program over $M$ of range $n$, let $u$ and $v$ be elements of $M$ such that
    $(\lambda, P, u, v)$ is of height $h$.
    Then there is an element $t$ of $M$ and a $\Delta$-compatible submask
    $\lambda'$ of $\lambda$ verifying
    $\length{\lambda'}_\Sigma \leq
     (2^h 6 l)^{2^h} \cdot \max\set{\length{\lambda}_\Sigma, 1}$
    such that any safe completion $w$ of $\lambda'$ verifies $u P(w) v = t$.
\end{lemma}

\begin{proof}
The proof goes by induction on the height $h$.

Let $\lambda$ be a $\Delta$-compatible mask of length $n$, let $P$ be a program
over $M$ for words of length $n$, let $u$ and $v$ be elements of $M$ such that
$(\lambda, P, u, v)$ is of height $h$, and assume that for any quadruplet
$(\lambda', P', u', v')$ strictly smaller than $(\lambda, P, u, v)$, the lemma
is verified. Consider the following conditions concerning the quadruplet
$(\lambda, P, u, v)$:
\begin{enumerate}[label=(\alph*)]
    \item\label{ctn:minimality_1}
	there does not exist any instruction $(x, f)$ of $P$ such that for some
	letter $a$ the submask $\lambda'$ of $\lambda$ formed by setting
	position $x$ to $a$ is $\Delta$-compatible and $f(a)$ is $R$-bad for $u$;
    \item\label{ctn:minimality_2}
	$v$ is not $R$-bad for $u$;
    \item\label{ctn:minimality_3}
	there does not exist any instruction $(x, f)$ of $P$ such that for some
	letter $a$ the submask $\lambda'$ of $\lambda$ formed by setting
	position $x$ to $a$ is
	$\Delta$-compatible and $f(a)$ is $L$-bad for $v$;
    \item\label{ctn:minimality_4}
	$u$ is not $L$-bad for $v$.
\end{enumerate}

We will now do a case analysis based on which of these conditions are violated
or not.
  \begin{itemize}[align=left]
    \item[Case 1:] condition~\ref{ctn:minimality_1} is violated. So there exists
	some instruction $(x, f)$ of $P$ such that for some letter $a$ the
	submask $\lambda'$ of $\lambda$ formed by setting position $x$ to $a$
	(if it wasn't already the case) is $\Delta$-compatible and $f(a)$ is
	$R$-bad for $u$. Let $i$ be the smallest number of such an instruction.

	Let $P'$ be the subprogram of $P$ until, and including, instruction $i - 1$. Let $w$ be
        a safe completion of $\lambda$. For any instruction $(y, g)$ of $P'$,
        as $y<i$, $g(w_y)$ cannot be $R$-bad for $u$, so $u \sim_R u g(w_y)$. Hence, by
	Lemma~\ref{lem:Green's_relations_DA}, $u \sim_R u P'(w)$ for all safe
	completions $w$ of $\lambda$.

	So, because $f(a)$ is $R$-bad for $u$, any safe completion $w$ of
	$\lambda'$, which is also a safe completion of $\lambda$, is such that
	$u \sim_R u P'(w) <_R u P'(w) f(a) \leq_R u P(w) v$ by
	Lemma~\ref{lem:Green's_relations_DA}, hence $u P'(w) <_J u P(w) v$ by
	Lemma~\ref{lem:Green's_relations_link}.
	So $(\lambda', P', u, 1) \prec (\lambda, P, u, v)$, therefore, by
	induction we get a $\Delta$-compatible submask $\lambda_1$ of $\lambda'$
	and a monoid element $t_1$ such that $u P'(w) = t_1$ for all safe
	completions $w$ of $\lambda_1$.

	Let $P''$ be the subprogram of $P$ starting from instruction $i + 1$.
	Notice that, since $u \sim_R t_1$ (by what we have proven just above),
	$u <_R t_1 f(a)$ (by Lemma~\ref{lem:Green's_relations_DA}) and
	$t_1 f(a) P''(w) v = u P'(w) f(a) P''(w) v = u P(w) v$ for all safe
	completions $w$ of $\lambda_1$.
	Hence, $(\lambda_1, P'', t_1 f(a), v)$ is strictly smaller than
	$(\lambda, P, u, v)$ and by induction we get a $\Delta$-compatible
	submask $\lambda_2$ of $\lambda_1$ and a monoid element $t$ such that
	$t_1f(a) P''(w) v = t$ for all safe completions $w$ of $\lambda_2$.

	Thus, any safe completion $w$ of $\lambda_2$ is such that
	\[
	    uP(w)v=uP'(w)f(a)P''(w)v=t_1f(a) P''(w) v =t
	    \displaypunct{.}
	\]
	Therefore $\lambda_2$ and $t$ form the desired couple of a
	$\Delta$-compatible submask of $\lambda$ and an element of $M$. We still
	have to show that $\length{\lambda_2}_\Sigma$ satisfies the desired
	upper bound.

	By induction, since $(\lambda', P', u, 1)$ is of height $h' \leq h - 1$,
	we have
	\[
	    \length{\lambda_1}_\Sigma \leq
	    (2^{h'} 6 l)^{2^{h'}} \cdot \length{\lambda'}_\Sigma
	    \leq
	    (2^{h - 1} 6 l)^{2^{h - 1}} \cdot
	    \length{\lambda'}_\Sigma
	    \displaypunct{.}
	\]
	Consequently, by induction again, as $(\lambda_1, P'', t_1 f(a), v)$ is
	of height $h'' \leq h - 1$, we have
	\begin{align*}
	    \length{\lambda_2}_\Sigma
	    & \leq (2^{h''} 6 l)^{2^{h''}} \cdot
		   \length{\lambda_1}_\Sigma\\
	    & \leq (2^{h - 1} 6 l)^{2^{h - 1}} \cdot
		   \length{\lambda_1}_\Sigma\\
	    & \leq (2^{h - 1} 6 l)^{2^{h - 1}} \cdot
		   (2^{h - 1} 6 l)^{2^{h - 1}} \cdot
		   \length{\lambda'}_\Sigma\\
	    & = (2^{h - 1} 6 l)^{2^h} \cdot
		\length{\lambda'}_\Sigma
	    \displaypunct{.}
	\end{align*}
	Moreover, it holds that
	$\length{\lambda'}_\Sigma \leq \length{\lambda}_\Sigma + 1$, so that
	\begin{align*}
	    \length{\lambda_2}_\Sigma
	    & \leq (2^{h - 1} 6 l)^{2^h} \cdot (\length{\lambda}_\Sigma + 1)\\
	    & \leq (2^{h - 1} 6 l)^{2^h} \cdot 2^{2^h} \cdot
		   \max\set{\length{\lambda}_\Sigma, 1}\\
	    & = (2^h 6 l)^{2^h} \cdot \max\set{\length{\lambda}_\Sigma, 1}
	    \displaypunct{.}
	\end{align*}

    \item[Case 2:] condition~\ref{ctn:minimality_1} is verified but
	condition~\ref{ctn:minimality_2} is violated, so $v$ is $R$-bad for $u$
	and Case~1 does not apply.

	Let $w$ be a safe completion of $\lambda$: for any instruction $(x, f)$
	of $P$, as the submask $\lambda'$ of $\lambda$ formed by setting
	position $x$ to $w_x$  is
	$\Delta$-compatible (by the fact that $\Delta$ is safe and $w$ is a safe
	completion of $\lambda$), $f(w_x)$ cannot be $R$-bad for $u$, otherwise
	condition~\ref{ctn:minimality_1} would be violated, so
	$u \sim_R u f(w_x)$.
	Hence, by Lemma~\ref{lem:Green's_relations_DA}, $u \sim_R u P(w)$ for
	all safe completions $w$ of $\lambda$.
	Notice then that $u \sim_R u P(w) <_R u P(w) v$ (by
	Lemma~\ref{lem:Green's_relations_DA}), hence $u P(w) <_J u P(w) v$
	(by Lemma~\ref{lem:Green's_relations_link}) for all safe completions $w$
	of $\lambda$. So $(\lambda, P, u, 1) \prec (\lambda, P, u, v)$,
	therefore we obtain by induction a monoid element $t_1$ and a $\Delta$-compatible
	submask $\lambda'$ of $\lambda$ such that $u P(w) = t_1$
	for all completions $w$ of $\lambda'$. If we set $t = t_1 v$, we get
	that any safe completion $w$ of $\lambda'$ is such that
	$u P(w) v = t_1 v = t$. Therefore $\lambda'$ and $t$ form the desired
	couple of a $\Delta$-compatible submask of $\lambda$ and an element of
	$M$.

	Moreover, by induction, since $(\lambda, P, u, 1)$ is of height
	$h' \leq h - 1$, we have
	\[
	    \length{\lambda'}_\Sigma \leq
	    (2^{h'} 6 l)^{2^{h'}} \cdot \max\set{\length{\lambda}_\Sigma, 1}
	    \leq (2^h 6 l)^{2^h} \cdot \max\set{\length{\lambda}_\Sigma, 1}
	    \displaypunct{,}
	\]
	the desired upper bound.

    \item[Case 3:] condition~\ref{ctn:minimality_3} is violated. So there exists
	some instruction $(x, f)$ of $P$ such that for some letter $a$ the
	submask $\lambda'$ of $\lambda$ formed by setting position $x$ to $a$
	(if it wasn't already the case) is $\Delta$-compatible and $f(a)$ is
	$L$-bad for $v$.

	We proceed as for Case~1 by symmetry.

    \item[Case 4:] condition~\ref{ctn:minimality_3} is verified but
	condition~\ref{ctn:minimality_4} is violated, so $u$ is $L$-bad for $v$
	and Case~3 does not apply.

	We proceed as for Case~2 by symmetry.

    \item[Case 5:] conditions~\ref{ctn:minimality_1}, \ref{ctn:minimality_2},
	\ref{ctn:minimality_3} and \ref{ctn:minimality_4} are verified.

	As it was in Case~2 and Case~4, using
	Lemma~\ref{lem:Green's_relations_DA}, the fact that
	condition~\ref{ctn:minimality_1} and condition~\ref{ctn:minimality_3}
	are verified implies that $u \sim_R u P'(w)$ and $v \sim_L P''(w) v$ for
	any prefix $P'$ of $P$, any suffix $P''$ of $P$ and all safe completions
	$w$ of $\lambda$. Moreover, since condition~\ref{ctn:minimality_2} and
	condition~\ref{ctn:minimality_4} are verified, by
	Lemma~\ref{lem:Green's_relations_DA}, we get that $u P(w) v \sim_R u$
	and $u P(w) v \sim_L v$ for all safe completions $w$ of $\lambda$. This
	implies that $(\lambda, P, u, v)$ is minimal for $\prec$ and that
	$h = 0$.

	Let $w_0$ be a completion of $\lambda$ that is in $\Delta^*$. Let
	$\lambda'$ be the submask of $\lambda$ fixing all free dangerous
	positions of $\lambda$ using $w_0$ and let $t = u P(w_0) v$. Then, for
	any completion $w$ of $\lambda'$, which is a safe completion of
	$\lambda$ by construction, we have that $u P(w) v \sim_R u \sim_R t$ and
	$u P(w) v \sim_L v \sim_L t$. Thus, $u P(w) v \sim_H t$ for any
	completion $w$ of $\lambda'$.
	As $M$ is aperiodic, this implies that $u P(w) v = t$ for all
	completions $w$ of $\lambda'$
	(see~\cite[Chapter 3, Proposition 4.2]{Books/Pin-1986}).
	Therefore $\lambda'$ and $t$ form the desired couple of a
	$\Delta$-compatible submask of $\lambda$ and an element of $M$.

	Now, since the number of free positions of $\lambda$ fixed in
	$\lambda'$, i.e. $\length{\lambda'}_\Sigma - \length{\lambda}_\Sigma$,
	is exactly the number of free dangerous positions in $\lambda$, and as a
	position in $\lambda$ is dangerous if it is within distance $2 l - 2$ of
	a fixed position or within distance $l - 1$ of the beginning or the end
	of $\lambda$, we have
	\begin{align*}
	    \length{\lambda'}_\Sigma
	    \leq 2 \cdot \length{\lambda}_\Sigma \cdot (2l - 2) +
		 2 \cdot (l - 1) + \length{\lambda}_\Sigma
	    & = \length{\lambda}_\Sigma \cdot (4l - 3) + 2l - 2\\
	    & \leq (2^0 6 l)^{2^0} \cdot \max\set{\length{\lambda}_\Sigma, 1}
	    \displaypunct{,}
	\end{align*}
	the desired upper bound.
\end{itemize}

This concludes the proof of the lemma.  
\end{proof}

Setting $\Delta = \set{c, ab}$ or $\Delta = \set{b, ab}$ with $\Sigma$ the
associated alphabet, when applying Lemma~\ref{lemma-R} with the trivial
$\Delta$-compatible mask $\lambda$ of length $n$ containing only free positions,
with $P$ some program over $M$ of range $n$ and with $u$ and $v$ equal to $1$,
the resulting mask $\lambda'$ has the property that we have an element $t$ of
$M$ such that $P(w)=t$ for any safe completion $w$ of $\lambda'$.
Since the mask $\lambda'$ is $\Delta$-compatible and has a number of fixed
positions upper-bounded by $(2^h 6 l)^{2^h}$ where $h$ is the height of
$(\lambda, P, u, v)$, itself upper-bounded by $2 \cdot \card{M}^2$, as long as
$n$ is big enough, we have a safe completion $w_0 \in \Delta^*$ and a safe
completion $w_1 \notin \Delta^*$. Hence, $P$ cannot be part of any sequence of
programs \pr-recognizing $\Delta^*$. This implies that
$(c + ab)^* \notin \Prog{M}$ and $(b + ab)^* \notin \Prog{M}$.
Finally, for any $k \in \N, k \geq 2$, we can prove that
$b^* ((a b^*)^k)^* \notin \Prog{M}$ by setting $\Delta = \set{a, b}$ and
completing the mask given by the lemma by setting the letters in such a way that
we have the right number of $a$ modulo $k$ in one case and not in the other
case.

This concludes the proof of Proposition~\ref{prop-main-PDA} because the argument
above holds for any monoid in $\FMVDA$.

\section{A fine hierarchy in \texorpdfstring{$\Prog{\FMVDA}$}{P(DA)}}
\label{sec:hier}

The definition of \pr-recognition by a sequence of programs over a monoid
given in Section~\ref{sec:Preliminaries} requires that for each $n$, the program
reading the entries of length $n$ has a length polynomial in $n$. In the case of
$\Prog{\FMVDA}$, the polynomial-length restriction is superfluous: any program
over a monoid in $\FMVDA$ is equivalent to one of polynomial length over the
same monoid~\cite{Tesson-Therien-2002a} (in the sense that they recognize the
same languages).
In this section, we show that this does not collapse further: in the case of
$\FMVDA$, programs of length $\Omicron(n^{k + 1})$ express strictly more than
those of length $\Omicron(n^k)$.

Following~\cite{Gavalda-Therien-2003}, we use an alternative definition of
the languages recognized by a monoid in $\FMVDA$. We define by induction a
hierarchy of classes of languages $\LVSUM[k]$, where $\LVSUM$ stands for
\emph{strongly unambiguous monomial}. A language $L$ is in $\LVSUM[0]$ if it is
of the form $A^*$ for some alphabet $A$. A language $L$ is in $\LVSUM[k]$ for
$k \in \N_{>0}$ if it is in $\LVSUM[k - 1]$ or $L = L_1 a L_2$ for some
languages $L_1 \in \LVSUM[i]$ and $L_2 \in \LVSUM[j]$ and some letter $a$ with
$i + j = k - 1$ such that no word of $L_1$ contains the letter $a$ or no word of
$L_2$ contains the letter $a$.

Gavaldà and Thérien stated without proof that a language $L$ is recognized by a
monoid in $\FMVDA$ iff there is a $k \in \N$ such that $L$ is a Boolean
combination of languages in $\LVSUM[k]$~\cite{Gavalda-Therien-2003}
(see~\cite[Theorem 4.1.9]{PhD_thesis/Grosshans} for a proof).
For each $k \in \N$, we denote by $\FMVDA[k]$ the variety of monoids generated
by the syntactic monoids of the Boolean combinations of languages in
$\LVSUM[k]$.
It can be checked that, for each $k$, $\FMVDA[k]$ forms a variety of monoids
recognizing precisely Boolean combinations of languages in $\LVSUM[k]$: this is
what we do in the first subsection.

In the two following subsections, we then give a fine program-length-based
hierarchy within $\Prog{\FMVDA}$ for this parametrization of $\FMVDA$.

\subsection{A parametrization of \texorpdfstring{$\FMVDA$}{DA}}

For each $k \in \N$, we denote by $\LVSUL[k]$ the class of regular languages
that are Boolean combinations of languages in $\LVSUM[k]$; it is a variety of
languages as shown just below. But as $\FMVDA[k]$ is the variety of monoids
generated by the syntactic monoids of the languages in $\LVSUL[k]$, by
Eilenberg's theorem, we know that, conversely, all the regular languages whose
syntactic monoids lie in $\FMVDA[k]$ are in $\LVSUL[k]$.

Back to the fact that $\LVSUL[k]$ is a variety of languages for any $k \in \N$.
Closure under Boolean operations is obvious by construction. Closure under
quotients and inverses of morphisms is respectively given by the following two
lemmas and by the fact that both quotients and inverses of morphisms commute
with Boolean operations.

Given a word $u$ over a given alphabet $\Sigma$, we will denote by
$\alphabet(u)$ the set of letters of $\Sigma$ that appear in $u$.

\begin{lemma}
\label{lem:SUP_k_closure_quotients}
    For all $k \in \N$, for all $L \in \LVSUM[k]$ over an alphabet $\Sigma$ and
    $u \in \Sigma^*$, $u^{-1} L$ and $L u^{-1}$ both are unions of languages in
    $\LVSUM[k]$ over $\Sigma$.
\end{lemma}

\begin{proof}
    We prove it by induction on $k$.

    \paragraph{Base case: $k = 0$.}
    Let $L \in \LVSUM[0]$ over an alphabet $\Sigma$ and $u \in \Sigma^*$.
    This means that $L = A^*$ for some $A \subseteq \Sigma$. We have two cases:
    either $\alphabet(u) \nsubseteq A$ and then
    $u^{-1} L = L u^{-1} = \emptyset$; or $\alphabet(u) \subseteq A$ and then
    $u^{-1} L = L u^{-1} = A^* = L$. So $u^{-1} L$ and $L u^{-1}$ both are
    unions of languages in $\LVSUM[0]$ over $\Sigma$.
    The base case is hence proved.

    \paragraph{Inductive step.}
    Let $k \in \N_{>0}$ and assume that the lemma is true for all
    $k' \in \N, k' < k$.

    Let $L \in \LVSUM[k]$ over an alphabet $\Sigma$ and $u \in \Sigma^*$.
    This means that either $L$ is in $\LVSUM[k - 1]$ and the lemma is proved by
    applying the inductive hypothesis directly for $L$ and $u$, or
    $L = L_1 a L_2$ for some languages $L_1 \in \LVSUM[i]$ and
    $L_2 \in \LVSUM[j]$ and some letter $a \in \Sigma$ with $i + j = k - 1$ and,
    either no word of $L_1$ contains the letter $a$ or no word of $L_2$ contains
    the letter $a$.
    We shall only treat the case in which $a$ does not appear in any of the
    words of $L_1$; the other case is treated symmetrically.

    There are again two cases to consider, depending on whether $a$ does appear
    in $u$ or not.

    If $a \notin \alphabet(u)$, then it is straightforward to check that
    $u^{-1} L = (u^{-1} L_1) a L_2$ and $L u^{-1} = L_1 a (L_2 u^{-1})$.
    By the inductive hypothesis, we get that $u^{-1} L_1$ is a union of
    languages in $\LVSUM[i]$ over $\Sigma$ and that $L_2 u^{-1}$ is a union of
    languages in $\LVSUM[j]$ over $\Sigma$. Moreover, it is direct to see that
    no word of $u^{-1} L_1$ contains the letter $a$.
    By distributivity of concatenation over union, we finally get that
    $u^{-1} L$ and $L u^{-1}$ both are unions of languages in $\LVSUM[k]$
    over~$\Sigma$.

    If $a \in \alphabet(u)$, then let $u = u_1 a u_2$ with
    $u_1, u_2 \in \Sigma^*$ and $a \notin \alphabet(u_1)$. It is again
    straightforward to see that
    \[
	u^{-1} L = \begin{cases}
		   {u_2}^{-1} L_2 & \text{if $u_1 \in L_1$}\\
		   \emptyset & \text{otherwise}
		   \end{cases}
    \]
    and
    \[
	L u^{-1} = L_1 a (L_2 u^{-1}) \cup
		   \begin{cases}
		   L_1 {u_1}^{-1} & \text{if $u_2 \in L_2$}\\
		   \emptyset & \text{otherwise}
		   \end{cases}
	\displaypunct{.}
    \]
    As before, by the inductive hypothesis, we get that $L_1 {u_1}^{-1}$ is a
    union of languages in $\LVSUM[i]$ over $\Sigma$ and that both
    ${u_2}^{-1} L_2$ and $L_2 u^{-1}$ are unions of languages in $\LVSUM[j]$
    over $\Sigma$. And, again, by distributivity of concatenation over union, we
    get that $u^{-1} L$ and $L u^{-1}$ both are a union of languages in
    $\LVSUM[k]$ over~$\Sigma$.

    This concludes the inductive step and therefore the proof of the
    lemma.
\end{proof}

\begin{lemma}
    For all $k \in \N$, for all $L \in \LVSUM[k]$ over an alphabet $\Sigma$ and
    $\varphi\colon \Gamma^* \to \Sigma^*$ a morphism where $\Gamma$ is another
    alphabet, $\varphi^{-1}(L)$ is a union of languages in $\LVSUM[k]$ over
    $\Gamma$.
\end{lemma}

\begin{proof}
    We prove it by induction on $k$.

    \paragraph{Base case: $k = 0$.}
    Let $L \in \LVSUM[0]$ over an alphabet $\Sigma$ and
    $\varphi\colon \Gamma^* \to \Sigma^*$ a morphism where $\Gamma$ is another
    alphabet. This means that $L = A^*$ for some $A \subseteq \Sigma$. It is
    straightforward to check that $\varphi^{-1}(L) = B^*$ where
    $B = \set{b \in \Gamma \mid \varphi(b) \in A^*}$. $B^*$ is certainly a union
    of languages in $\LVSUM[0]$ over $\Sigma$.
    The base case is hence proved.

    \paragraph{Inductive step.}
    Let $k \in \N_{>0}$ and assume that the lemma is true for all
    $k' \in \N, k' < k$.
	    
    Let $L \in \LVSUM[k]$ over an alphabet $\Sigma$ and
    $\varphi\colon \Gamma^* \to \Sigma^*$ a morphism where $\Gamma$ is another
    alphabet. This means that either $L$ is in $\LVSUM[k - 1]$ and the lemma is
    proved by applying the inductive hypothesis directly for $L$ and $\varphi$,
    or $L = L_1 a L_2$ for some languages $L_1 \in \LVSUM[i]$ and
    $L_2 \in \LVSUM[j]$ and some letter $a \in \Sigma$ with $i + j = k - 1$ and,
    either no word of $L_1$ contains the letter $a$ or no word of $L_2$ contains
    the letter $a$. We shall only treat the case in which $a$ does not appear in
    any of the words of $L_1$; the other case is treated symmetrically.

    Let us define $B = \set{b \in \Gamma \mid a \in \alphabet(\varphi(b))}$ as
    the set of letters of $\Gamma$ whose image word by $\varphi$ contains the
    letter $a$. For each $b \in B$, we shall also let
    $\varphi(b) = u_{b, 1} a u_{b, 2}$ with $u_{b, 1}, u_{b, 2} \in \Sigma^*$
    and $a \notin \alphabet(u_{b, 1})$.
    It is not too difficult to see that we then have
    \[
	\varphi^{-1}(L) =
	\bigcup_{b \in B} \varphi^{-1}(L_1 {u_{b, 1}}^{-1}) b
			  \varphi^{-1}({u_{b, 2}}^{-1} L_2)
	\displaypunct{.}
    \]
    By the inductive hypothesis, by Lemma~\ref{lem:SUP_k_closure_quotients} and
    by the fact that inverses of morphisms commute with unions, we get that
    $\varphi^{-1}(L_1 {u_{b, 1}}^{-1})$ is a union of languages in $\LVSUM[i]$
    over $\Gamma$ and that $\varphi^{-1}({u_{b, 2}}^{-1} L_2)$ is a union of
    languages in $\LVSUM[j]$ over $\Gamma$. Moreover, it is direct to see that
    no word of $\varphi^{-1}(L_1 {u_{b, 1}}^{-1})$ contains the letter $b$ for
    all $b \in B$.
    By distributivity of concatenation over union, we finally get that
    $\varphi^{-1}(L)$ is a union of languages in $\LVSUM[k]$ over~$\Gamma$.

    This concludes the inductive step and therefore the proof of the lemma.
\end{proof}

\subsection{Strict hierarchy}

For each $k$ we show there exists a language $L_k \subseteq \set{0, 1}^*$ that
can be recognized by a sequence of programs of length $\Omicron(n^k)$ over a
monoid $M_k$ in $\FMVDA[k]$ but cannot be recognized by any sequence of programs
of length $\Omicron(n^{k - 1})$ over any monoid in $\FMVDA$.

For a given $k \in \N_{>0}$, the language $L_k$ expresses a property of the
first $k$ occurrences of $1$ in the input word. To define $L_k$ we say that $S$
is a \emph{$k$-set over $n$} for some $n \in \N$ if $S$ is a set where each
element is an ordered tuple of $k$ distinct elements of $[n]$.
For any sequence $\Delta=(S_n)_{n\in\N}$ of $k$-sets over $n$, we set
$L_\Delta=\bigcup_{n\in\N} K_{n,S_n}$, where for each $n \in \N$, $K_{n,S_n}$ is
the set of words over $\set{0, 1}$ of length $n$ such that for each of them, it
contains at least $k$ occurrences of $1$ and the ordered $k$-tuple of the
positions of the first $k$ occurrences of $1$ belongs to $S_n$.

On the one hand, we show that for all $k$ there is a monoid $M_k$ in $\FMVDA[k]$
such that for all $\Delta$ the language $L_\Delta$ is recognized by a sequence
of programs over $M_k$ of length $\Omicron(n^k)$. The proof is done by an
inductive argument on $k$.

On the other hand, we show that for all $k$ there is a $\Delta$ such that for
any finite monoid $M$ and any sequence of programs $(P_n)_{n \in \N}$ over $M$
of length $\Omicron(n^{k-1})$, $L_\Delta$ is not recognized by
$(P_n)_{n \in \N}$.
This is done using a counting argument: for some monoid size $i$, for $n$ big
enough, the number of languages in $\set{0, 1}^n$ recognized by a program over
some monoid of size $i$ of length at most $\alpha \cdot n^{k - 1}$ for $\alpha$
some constant is upper-bounded by a number that turns out to be asymptotically
smaller than the number of different possible $K_{n, S_n}$.

\paragraph*{Upper bound.}
We start with the upper bound. Notice that for some $k \in \N_{>0}$ and
$\Delta=(S_n)_{n\in\N}$, the language of words of length $n$ of $L_\Delta$ is
exactly $K_{n,S_n}$. Hence the fact that $L_\Delta$ can be recognized by a
sequence of programs over a monoid in $\FMVDA[k]$ of length $\Omicron(n^k)$ is a
consequence of the following proposition.

\begin{proposition}\label{prop-upper-bound}
    For all $k \in \N_{>0}$ there is a monoid $M_k \in \FMVDA[k]$ such that for
    all $n \in \N$ and all $k$-sets $S_n$ over $n$, the language $K_{n, S_n}$ is
    recognized by a program over $M_k$ of length at most $4 n^k$.
\end{proposition}

\begin{proof}
    We first define by induction on $k$ a family of languages $Z_k$ over the
    alphabet $Y_k = \set{\bot_l, \top_l \mid 1 \leq l \leq k}$.  For $k = 0$,
    $Z_0$ is $\set{\emptyword}$.  For $k >0$, $Z_k$ is the set of words
    containing $\top_k$ and such that the first occurrence of $\top_k$ has no
    $\bot_k$ to its left, and the sequence between the first occurrence of
    $\top_k$ and the first occurrence of $\bot_k$ or $\top_k$ to its right, or
    the end of the word if there is no such letter, belongs to $Z_{k-1}$. A
    simple induction on $k$ shows that $Z_k$ is defined by the
    following expression
    \[Y_{k-1}^* \top_{k} Y_{k-2}^* \top_{k-1} \cdots Y_1^* \top_2 \top_1 Y_k^*\]
    and therefore it is in $\LVSUM[k]$ and its syntactic
    monoid $M_k$ is in $\FMVDA[k]$.

    Fix $n$. If $n = 0$, the proposition follows trivially, otherwise, we define
    by induction on $k$ a program $P_k(i, S)$ for every $k$-set $S$ over $n$ and
    every $1 \leq i \leq n + 1$ that will for the moment output elements of
    $Y_k \cup \set{\emptyword}$ instead of outputting elements of $M_k$.

    For any $k >0$, $1 \leq j \leq n$ and $S$ a $k$-set over $n$, let $f_{j, S}$
    be the function with $f_{j, S}(0) = \emptyword$ and $f_{j, S}(1) = \top_k$
    if $j$ is the first element of some ordered $k$-tuple of $S$,
    $f_{j, S}(1) = \bot_k$ otherwise.  We also let $g_k$ be the function with
    $g_k(0) = \emptyword$ and $g_k(1) = \bot_k$.  If $S$ is a $k$-set over $n$
    and $1 \leq j \leq n$ then $S|j$ denotes the $(k - 1)$-set over $n$
    containing the ordered $(k - 1)$-tuples $\bar t$ such that
    $(j, \bar t) \in S$.

    For $k>0$, $1 \leq i \leq n + 1$ and $S$ a $k$-set over $n$, the
    program $P_k(i, S)$ is the following sequence of instructions:
    \[
	(i, f_{i, S}) P_{k - 1} (i + 1, S|i) (i, g_k) \cdots
	(n, f_{n, S}) P_{k - 1} (n + 1, S|n) (n, g_k).
    \]

    In other words, the program guesses the first occurrence $j\geq i$ of $1$,
    returns $\bot_k$ or $\top_k$ depending on whether it is the first element of
    an ordered $k$-tuple in $S$, and then proceeds for the next occurrences of
    $1$ by induction.

    For $k = 0$, $1 \leq i \leq n + 1$ and $S$ a $0$-set over $n$ (that is empty
    or contains $\emptyword$, the only ordered $0$-tuple of elements of $[n]$),
    the program $P_0(i, S)$ is the empty program $\emptyword$.
    
    A simple computation shows that for any $k \in \N_{>0}$,
    $1 \leq i \leq n + 1$ and $S$ a $k$-set over $n$, the number of instructions
    in $P_k(i, S)$ is at most $4 n^k$.

    A simple induction on $k$ shows that when running on a word
    $w \in \set{0, 1}^n$, for any $k \in \N_{>0}$, $1 \leq i \leq n + 1$ and $S$
    a $k$-set over $n$, $P_k(i, S)$ returns a word in $Z_k$ iff the ordered
    $k$-tuple of the positions of the first $k$ occurrences of $1$ starting at
    position $i$ in $w$ exists and is an element of $S$.

    For any $k>0$ and $S_n$ a $k$-set over $n$, it remains to apply
    the syntactic morphism of $Z_k$ to the output of the functions in the
    instructions of $P_k(1, S_n)$ to get a program over $M_k$ of length at most
    $4 n^k$ recognizing $K_{n, S_n}$.
\end{proof}

\paragraph*{Lower bound.}
The following claim is a simple counting argument.

\begin{claim}\label{claim-monoid-fixed}
    For all $i \in \N_{>0}$ and $n \in \N$, the number of languages in
    $\set{0, 1}^n$ recognized by programs over a monoid of size $i$, reading
    inputs of length $n$ over the alphabet $\set{0, 1}$, with at most $l \in \N$
    instructions, is bounded by $i^{i^2} 2^i \cdot (n \cdot i^2)^l$.
\end{claim}

\begin{proof}
    Fix a monoid $M$ of size $i$.
    Since a program over $M$ of range $n$ with less than $l$ instructions can
    always be completed into such a program with exactly $l$ instructions
    recognizing the same languages in $\set{0, 1}^n$ (using the identity of
    $M$), we only consider programs with exactly $l$ instructions.
    As $\Sigma = \set{0, 1}$, there are $n \cdot i^2$ choices for each of the
    $l$ instructions of a range $n$ program over $M$ reading inputs in
    $\set{0, 1}^*$. Such a program can recognize at most $2^i$ different
    languages in $\set{0, 1}^n$. Hence, the number of languages in
    $\set{0, 1}^n$ recognized by programs over $M$ of length at most $l$ is at
    most $2^i \cdot (n \cdot i^2)^l$.
    The result follows from the facts that there are at most $i^{i^2}$
    isomorphism classes of monoids of size $i$ and that two isomorphic monoids
    allow to recognize the same languages in $\set{0, 1}^n$ through programs.
\end{proof}

If for some $k \in \N_{>0}$ and $1 \leq i \leq \alpha$, $\alpha \in \N_{>0}$, we
apply Claim~\ref{claim-monoid-fixed} for all $n \in \N$,
$l = \alpha \cdot n^{k - 1}$, we get a number $\mu_i(n)$ of languages
upper-bounded by $n^{\Omicron(n^{k - 1})}$, which is asymptotically strictly
smaller than the number of distinct $K_{n, S_n}$, which is $2^{\binom{n}{k}}$,
i.e. $\mu_i(n)$ is in $\omicron\bigl(2^{\binom{n}{k}}\bigr)$.

Hence, for all $j \in \N_{>0}$, there exist an $n_j \in \N$ and $T_j$ a $k$-set
over $n_j$ such that no program over a monoid of size $1 \leq i \leq j$, of
range $n_j$ and of length at most $j \cdot n^{k - 1}$ recognizes $K_{n_j, T_j}$.
Moreover, we can assume without loss of generality that the sequence
$(n_j)_{j \in \N_{>0}}$ is increasing.
Let $\Delta = (S_n)_{n \in \N}$ be such that $S_{n_j} = T_j$ for all
$j \in \N_{>0}$ and $S_n = \emptyset$ for any $n \in \N$ verifying that it is
not equal to any $n_j$ for $j \in \N_{>0}$. We show that no
sequence of programs over a finite monoid of length $\Omicron(n^{k - 1})$ can
recognize $L_\Delta$. If this were the case, then let $i$ be the size of the
monoid. Let $j \geq i$ be such that for any $n \in \N$, the $n$-th program has
length at most $j \cdot n^{k - 1}$. But, by construction, we know that there
does not exist any such program of range $n_j$ recognizing $K_{n_j,T_j}$, a
contradiction.

This implies the following hierarchy, where $\Prog{\V, s(n)}$ for some variety
of monoids $\V$ and a function $s\colon \N \to \N$ denotes the class of
languages recognizable by a sequence of programs of length $\Omicron(s(n))$:

\begin{proposition}
    For all $k \in \N$,
    $\Prog{\FMVDA, n^k} \subsetneq \Prog{\FMVDA, n^{k + 1}}$.
    More precisely, for all $k \in \N$ and
    $d \in \N, d \leq \max\set{k - 1, 0}$,
    $\Prog{\FMVDA[k], n^d} \subsetneq \Prog{\FMVDA[k], n^{d + 1}}$.
\end{proposition}

To prove this proposition, we use two facts. First, that for all $k \in \N$ and
all $d \in \N, d \leq \max\set{k - 1, 0}$, any monoid from $\FMVDA[d]$ is also a
monoid from $\FMVDA[k]$. And second, that
$a^* \in \Prog{\FMVDA[0], n} \setminus \Prog{\FMVDA[0], 1}$ simply because any
program over some finite monoid of range $n$ for $n \in \N$ recognizing $a^n$
must have at least $n$ instructions, one for each input letter.

\subsection{Collapse}

Tesson and Thérien showed that any program over a monoid $M$ in $\FMVDA$ is
equivalent to one of polynomial length~\cite{Tesson-Therien-2002a}.
We now show that if we further assume that $M$ is in $\FMVDA[k]$ then the length
can be assumed to be $\Omicron(n^{\max\set{k, 1}})$. 

\begin{proposition}
\label{ptn:Collapse}
    Let $k \geq 0$. Let $M \in \FMVDA[k]$.
    Then any program over $M$ is equivalent to a program over $M$ of length
    $\Omicron(n^{\max\set{k, 1}})$ which is a subprogram of the initial one.
\end{proposition}

For each possible acceptance set, an input word
to the program is accepted if and only if the word over the alphabet $M$
produced by the program belongs to some fixed Boolean combination of languages
in $\LVSUM[k]$. The idea is then just to keep enough instructions so that
membership of the produced word over $M$ in each of these languages does not
change.

Recall that if $P$ is a program over some monoid $M$ of range $n$, then $P(w)$
denotes the element of $M$ resulting from the execution of the program $P$ on
$w$.
It will be convenient here to also work with the word over $M$ resulting from
the sequence of executions of each instruction of $P$ on $w$. We denote this
word by $\EP(w)$.

The result is a consequence of the following lemma and the fact that for any
acceptance set $F \subseteq M$, a word $w \in \Sigma^n$ (where $\Sigma$ is the
input alphabet) is accepted iff $\EP(w) \in L$ where $L$ is a language in
$\LVSUL[k]$, a Boolean combination of languages in $\LVSUM[k]$.

\begin{lemma}
\label{lem:Hierarchy-Program_compression_unambiguous_monomial}
    Let $\Sigma$ be an alphabet, $M$ a finite monoid, and $n, k$ natural
    numbers.

    For any program $P$ over $M$ of range $n$ and any language $K$ over $M$ in
    $\LVSUM[k]$, there exists a subprogram $Q$ of $P$ of length
    $\Omicron(n^{\max\set{k, 1}})$ such that for any subprogram $Q'$ of $P$ that
    has $Q$ as a subprogram, we have for all words $w$ over $\Sigma$ of length
    $n$:
    \[
	\EP(w) \in K \Leftrightarrow \EP[Q'](w) \in K
	\displaypunct{.}
    \]
\end{lemma}

\begin{proof}
    A program $P$ over $M$ of range $n$ is a finite sequence $(p_i, f_i)$ of
    instructions where each $p_i$ is a positive natural number which is at most
    $n$ and each $f_i$ is a function from $\Sigma$ to $M$. We denote by $l$ the
    number of instructions of $P$. For each set $I\subseteq [l]$ we denote by
    $P[I]$ the subprogram of $P$ consisting of the subsequence of instructions
    of $P$ obtained after removing all instructions whose index is not in $I$.
    In particular, $P[1, m]$ denotes the initial sequence of instructions of
    $P$, until instruction number $m$.

    We prove the lemma by induction on $k$.

    The intuition behind the proof for a program $P$ on inputs of length $n$ and
    some $K_1 \gamma K_2 \in \LVSUM[k]$ when $k \geq 2$ is as follows.
    We assume that $K_1$ does not contain any word with the letter $\gamma$, the
    other case is done symmetrically.
    Consider the subset of all indices $I_\gamma \subseteq [l]$ that correspond,
    for a fixed letter $a$ and a fixed position $p$ in the input, to the first
    instruction of $P$ that would output the element $\gamma$ when
    reading $a$ at position $p$.
    We then have that, given some $w$ as input,
    $\EP(w) \in K_1 \gamma K_2$ if and only if there exists
    $i \in I_\gamma$ verifying that the element at position $i$ of
    $\EP(w)$ is $\gamma$, $\EP[{P[1, i - 1]}](w) \in K_1$ and
    $\EP[{P[i + 1, l]}](w) \in K_2$.
    The idea is then that if we set $I$ to contain $I_\gamma$ as well as all
    indices obtained by induction for $P[1, i - 1]$ and $K_1$ and for
    $P[i + 1, l]$ and $K_2$, we would have that for all $w$,
    $\EP(w) \in K_1 \gamma K_2$ if and only if
    $\EP[{P[I]}](w) \in K_1 \gamma K_2$, that is $\EP(w)$ where only the
    elements at indices in $I$ have been kept.

    The intuition behind the proof when $k < 2$ is essentially the same, but
    without induction.

    We now spell out the details of the proof, starting with the inductive step.
    
    \paragraph{Inductive step.}
    Let $k \geq 2$ and assume the lemma proved for all $k' < k$. Let $n$ be a
    natural number, $P$ a program over $M$ of range $n$ and length $l$ and any
    language $K$ over $M$ in $\LVSUM[k]$.
    If $K \in \LVSUM[k - 1]$, by the inductive hypothesis, we are done.
    Otherwise, by definition, $K = K_1 \gamma K_2$ for $\gamma \in M$ and some
    languages $K_1 \in \LVSUM_{k_1}$ and $K_2 \in \LVSUM_{k_2}$ over $M$ with
    $k_1 + k_2 = k - 1$.
    Moreover either $\gamma$ does not occur in any of the words of $K_1$ or it
    does not occur in any of the words of $K_2$. We only treat the case where
    $\gamma$ does not appear in any of the words in $K_1$. The other case is
    treated similarly by symmetry.

    Observe that when $n = 0$, we necessarily have $P = \emptyword$, so that the
    lemma is trivially proven in that case. So we now assume $n > 0$.

    For each $1 \leq p \leq n$ and each $a \in \Sigma$ consider within the
    sequence of instructions of $P$ the first instruction of the form $(p, f)$
    with $f(a) = \gamma$, if it exists. We let $I_\gamma$ be the set of indices
    of these instructions for all $a$ and $p$. Notice that the size of
    $I_\gamma$ is in $\Omicron(n)$.

    For all $i \in I_\gamma$, we let $J_{i, 1}$ be the set of indices of the
    instructions within $P[1, i-1]$ appearing in its subprogram obtained by
    induction for $P[1, i-1]$ and $K_1$, and $J_{i, 2}$ be the same for
    $P[i+1, l]$ and $K_2$.

    We now let $I$ be the union of $I_\gamma$ and $J_{i, 1}$ and
    $J_{i, 2}' = \set{j + i \mid j \in J_{i,2}}$ for all $i\in I_\gamma$. We
    claim that $Q = P[I]$ has the desired properties.

    First notice that by induction the sizes of $J_{i, 1}$ and $J_{i, 2}'$ for
    all $i \in I_\gamma$ are in
    $\Omicron(n^{\max\set{k - 1, 1}}) \allowbreak = \Omicron(n^{k - 1})$ and
    because the size of $I_\gamma$ is linear in $n$, the size of $I$ is in
    $\Omicron(n^k) = \Omicron(n^{\max\set{k, 1}})$ as required.

    Let $Q'$ be a subprogram of $P$ that has $Q$ as a subprogram: it means that
    there exists some set $I' \subseteq [l]$ containing $I$ such that
    $Q' = P[I']$.

    Now take $w \in \Sigma^n$.

    Assume now that $\EP(w) \in K$. Let $i$ be the position in $\EP(w)$ of label
    $\gamma$ witnessing the membership in $K$. Let $(p_i, f_i)$ be the
    corresponding instruction of $P$. In particular we have that
    $f_i(w_{p_i}) = \gamma$. Because $\gamma$ does not occur in any word of
    $K_1$, for all $j < i$ such that $p_j = p_i$ we cannot have
    $f_j(w_{p_j}) = \gamma$. Hence $i \in I_\gamma$. By induction we have that
    $\EP[{P[1, i - 1][J]}](w) \in K_1$ for any set $J \subseteq [i -1]$
    containing $J_{i, 1}$ and $\EP[{P[i + 1, l][J]}](w) \in K_2$ for any set
    $J \subseteq [l - i]$ containing $J_{i, 2}$.
    Hence, if we set $I_1' = \set{j \in I' \mid j < i}$ as the subset of $I'$ of
    elements less than $i$ and $I_2' = \set{j - i \in I' \mid j > i}$ as the
    subset of $I'$ of elements greater than $i$ translated by $-i$, we have
    \[
	\EP[P[I']](w) =
	\EP[{P[1, i - 1][I_1']}](w) \gamma \EP[{P[i + 1, l][I_2']}](w) \in
	K_1 \gamma K_2 = K
    \]
    as desired.

    Assume finally that $\EP[P[I']](w) \in K$. Let $i$ be the index in $I'$
    whose instruction provides the letter $\gamma$ witnessing the fact that
    $\EP[P[I']](w) \in K$. This means that if we set
    $I_1' = \set{j \in I' \mid j < i}$ as the subset of $I'$ of elements less
    than $i$ and $I_2' = \set{j - i \in I' \mid j > i}$ as the subset of $I'$ of
    elements greater than $i$ translated by $-i$, we have
    $\EP[P[I']](w) =
     \EP[{P[1, i - 1][I_1']}](w) \gamma \EP[{P[i + 1, l][I_2']}](w)$
    with $\EP[{P[1, i - 1][I_1']}](w) \in K_1$ and
    $\EP[{P[i + 1, l][I_2']}](w) \in K_2$.
    If $i\in I_\gamma$, then it means that $I_1' \subseteq [i - 1]$ contains
    $J_{i, 1}$ and that $I_2' \subseteq [l - i]$ contains $J_{i, 2}$ by
    construction, so that, by induction,
    \[
	\EP[P](w) = \EP[{P[1, i - 1]}](w) \gamma \EP[{P[i + 1, l]}](w) \in
	K_1 \gamma K_2 = K
	\displaypunct{.}
    \]
    If not this shows that there is an instruction $(p_j, f_j)$ with $j < i$,
    $j \in I'$, $p_j = p_i$ and $f_j(w_{p_j}) = \gamma$. But that would
    contradict the fact that $\gamma$ cannot occur in $K_1$.
    So we have $\EP(w) \in K$ as desired.
    
    \paragraph{Base case.}
    There are two subcases to consider.
    
    \subparagraph{Subcase $k = 1$.}
    Let $n$ be a natural number, $P$ a program over $M$ of range $n$ and length
    $l$ and any language $K$ over $M$ in $\LVSUM[1]$.

    If $K \in \LVSUM[0]$, we can conclude by referring to the subcase $k = 0$.

    Otherwise $K = A_1^* \gamma A_2^*$ for $\gamma \in M$ and some alphabets
    $A_1 \subseteq M$ and $A_2 \subseteq M$. Moreover either $\gamma \notin A_1$
    or $\gamma \notin A_2$. We only treat the case where $\gamma$ does not
    belong to $A_1$, the other case is treated similarly by symmetry.

    We use the same idea as in the inductive step.

    Observe that when $n = 0$, we necessarily have $P = \emptyword$, so that the
    lemma is trivially proven in that case. So we now assume $n > 0$.

    For each $1 \leq p \leq n$, each $\alpha \in M$ and $a \in \Sigma$ consider
    within the sequence of instructions of $P$ the first and last instruction of
    the form $(p, f)$ with $f(a) = \alpha$, if they exist.
    We let $I$ be the set of indices of these instructions for all $a, \alpha$
    and $p$. Notice that the size of $I$ is in
    $\Omicron(n) = \Omicron(n^{\max\set{k, 1}})$.

    We claim that $Q = P[I]$ has the desired properties. We just showed that it
    has the required length.

    Let $Q'$ be a subprogram of $P$ that has $Q$ as a subprogram: it means that
    there exists some set $I' \subseteq [l]$ containing $I$ such that
    $Q' = P[I']$.
    
    Take $w \in \Sigma^n$.

    Assume now that $\EP(w) \in K$. Let $i$ be the position in $\EP(w)$ of label
    $\gamma$ witnessing the membership in $K$. Let $(p_i, f_i)$ be the
    corresponding instruction of $P$. In particular we have that
    $f_i(w_{p_i}) = \gamma$ and this is the $\gamma$ witnessing the membership
    in $K$. Because $\gamma \notin A_1$, for all $j < i$ such that $p_j = p_i$
    we cannot have $f_j(w_{p_j}) = \gamma$. Hence $i \in I \subseteq I'$. From
    $\EP[P[1, i - 1]](w) \in A_1^*$ and $\EP[P[i + 1, l]](w) \in A_2^*$ it
    follows that $\EP[P[I' \cap \intinterval{1}{i - 1}]](w) \in A_1^*$ and
    $\EP[P[I' \cap \intinterval{i + 1}{l}]](w) \in A_2^*$, showing that
    $\EP[P[I']](w) =
     \EP[P[I' \cap \intinterval{1}{i - 1}]](w) \gamma
     \EP[P[I' \cap \intinterval{i + 1}{l}]](w) \in K$
    as desired.

    Assume finally that $\EP[P[I']](w) \in K$. Let $i$ be the index in $I'$
    whose instruction provides the letter $\gamma$ witnessing the fact that
    $\EP[P[I']](w) \in K$. This means that
    $\EP[P[I' \cap \intinterval{1}{i - 1}]](w) \in A_1^*$ and
    $\EP[P[I' \cap \intinterval{i + 1}{l}]](w) \in A_2^*$. If there is an
    instruction $(p_j, f_j)$, with $j < i$ and $f_j(w_{p_j}) \notin A_1$ then
    either $j \in I'$ and we get a direct contradiction with the fact that
    $\EP[P[I' \cap \intinterval{1}{i - 1}]](w) \in A_1^*$, or $j \notin I'$ and
    we get a smaller $j' \in I \subseteq I'$ with the same property,
    contradicting again the fact that
    $\EP[P[I' \cap \intinterval{1}{i - 1}]](w) \in A_1^*$. Hence for all
    $j < i$, $f_j(w_{p_j}) \in A_1$. By symmetry we have that for all $j > i$,
    $f_j(w_{p_j}) \in A_2$, showing that $\EP(w) \in A_1^* \gamma A_2^* = K$ as
    desired.

    \subparagraph{Subcase $k = 0$.}
    Let $n$ be a natural number, $P$ a program over $M$ of range $n$ and length
    $l$ and any language $K$ over $M$ in $\LVSUM[0]$.

    Then $K = A^*$ for some alphabet $A \subseteq M$.

    We again use the same idea as before.

    Observe that when $n = 0$, we necessarily have $P = \emptyword$, so that the
    lemma is trivially proven in that case. So we now assume $n > 0$.

    For each $1 \leq p \leq n$, each $\alpha \in M$ and $a \in \Sigma$ consider
    within the sequence of instructions of $P$ the first instruction of the form
    $(p, f)$ with $f(a) = \alpha$, if it exists.
    We let $I$ be the set of indices of these instructions for all $a, \alpha$
    and $p$. Notice that the size of $I$ is in
    $\Omicron(n) = \Omicron(n^{\max\set{k, 1}})$.

    We claim that $Q = P[I]$ has the desired properties. We just showed that it
    has the required length.

    Let $Q'$ be a subprogram of $P$ that has $Q$ as a subprogram: it means that
    there exists some set $I' \subseteq [l]$ containing $I$ such that
    $Q' = P[I']$.
    
    Take $w \in \Sigma^n$.

    Assume now that $\EP(w) \in K$. As $\EP[P[I']](w)$ is a subword of $\EP(w)$,
    it follows directly that $\EP[P[I']](w) \in A^* = K$ as desired.

    Assume finally that $\EP[P[I']](w) \in K$. If there is an instruction
    $(p_j, f_j)$, with $j \in [l]$ and $f_j(w_{p_j}) \notin A$ then either
    $j \in I'$ and we get a direct contradiction with the fact that
    $\EP[P[I]](w) \in A^* = K$, or $j \notin I'$ and we get a smaller
    $j' \in I \subseteq I'$ with the same property, contradicting again the fact
    that $\EP[P[I']](w) \in A^* = K$. Hence for all $j \in [l]$,
    $f_j(w_{p_j}) \in A$, showing that $\EP(w) \in A^* = K$ as desired.
\end{proof}

\section{Conclusion}

We introduced a notion of tameness, particularly relevant to the analysis of programs over monoids from ``small'' varieties.
The main source of interest in tameness is Proposition~\ref{prop-tame-quasi-essentially}, stating that a variety of monoids $\V$ is tame if and only if the class of regular languages \pr-recognized by
programs over monoids from $\V$ is included in the class $\DLang{\StVQuasi\StVEsntl\V}$.
A first question that arises is for which $\V$ those two classes of regular
languages are equal. We could not rule out the possibility that for some tame
non-trivial $\V$,
$\DLang{\StVQuasi\StVEsntl\V}\setminus \Prog{\V}\neq \emptyset$.  We
conjecture that if $\V$ is local, abusing notation, $\StVQuasi\StVEsntl\V = \StVEsntl\V\FSVsdprod\Mod$, by analogy with $\StVQuasi\V$ equating $\V\FSVsdprod\Mod$ in that case; as
$\DLang{\StVEsntl\V\FSVsdprod\Mod}\subseteq \Prog{\V}$ holds unconditionally
(because $\V$ cannot be trivial if it is local~\cite[p.~134]{Tilson-1987}), under our conjecture
$\Prog{\V}\cap \Reg = \DLang{\StVQuasi\StVEsntl\V}$ would hold for local tame
varieties $\V$ (this is Conjecture~\ref{cjt:Locality_tameness}).

Concretely, we have obtained the technical result that $\FMVDA$ is a tame
variety using semigroup-theoretic arguments.
We have given $\FMVA$ and
$\FMVCom$ as further examples of tame varieties.
Our proof that $\FMVA$ is tame needed the fact that
$\FLang{MOD_m}\notin \AC[0]$ for all $m \geq 2$, so it would be
interesting to prove $\FMVA$ tame ``purely algebraically'', independently from
the known combinatorial
arguments~\cite{Ajtai-1983, Furst-Saxe-Sipser-1984, Hastad-1986} and those
based on approximating circuits by polynomials over some finite
field~\cite{Razborov-1987, Smolensky-1987}.
But tameness of $\FMVA$ is actually equivalent to
$\FLang{MOD_m}\notin \AC[0]$ for all $m \geq 2$ by
Proposition~\ref{prop-tame-quasi-essentially}, thus confronting us with the
challenging task to reprove significant circuit complexity results by relying
mainly on new semigroup-theoretic arguments. Such a breakthrough is still to be
made.

By contrast, we have shown that 
$\FMVJ$ is not tame. So programs over monoids from $\FMVJ$ \pr-recognize ``more regular languages than expected''. A natural question to ask is what these regular languages in $\Prog{\FMVJ}$ are.  Partial results in that direction were obtained in~\cite{Grosshans-2020}.

To conclude we should add, in fairness, that the progress reported here does not in any obvious way
bring us closer to major $\NC[1]$ complexity subclasses separations.  Our concrete contributions here largely concern $\Prog{\FMVDA}$ and $\Prog{\FMVJ}$, classes that are well within $\AC[0]$.
But this work does uncover new ways in which a program can or cannot circumvent the
limitations imposed by the underlying monoid algebraic structure available to
it.

\section*{Acknowledgments}
\noindent The authors would like to thank the anonymous referees for their
detailed reports as well as for their constructive criticism and the many
suggestions and comments they made. This helped us to really improve our paper.

\bibliographystyle{alphaurl}
\bibliography{Bibliography}

\end{document}